%% file: main.tex
\documentclass[11pt]{article} 

\usepackage[utf8]{inputenc} 

\usepackage{cancel}

\usepackage{geometry} 
\usepackage{graphicx} 
\usepackage[english]{babel}

\usepackage{booktabs} 
\usepackage{array} 
\usepackage{paralist} 
\usepackage{verbatim} 
\usepackage{subfig} 
\usepackage{amssymb}
\usepackage{bbm}

\usepackage{amsmath}
\usepackage{amsthm}
\usepackage{makecell}

\usepackage{epigraph}

\usepackage{hyperref}
\hypersetup{
    colorlinks=true,
    citecolor=blue
}

\usepackage{pdflscape}
\usepackage{natbib}
\usepackage{algorithm2e}

\usepackage{wrapfig}
\usepackage{subfig}
\usepackage[capitalize,noabbrev]{cleveref}
\usepackage{authblk}
\usepackage{changepage}

\usepackage{tikz-cd}
\usepackage{xcolor}

\input{macros}

\newtheorem{theorem}{Theorem}[section]
\newtheorem{definition}[theorem]{Definition}

\newtheorem{proposition}[theorem]{Proposition}

\newtheorem{corollary}[theorem]{Corollary}
\newtheorem{lemma}[theorem]{Lemma}
\theoremstyle{definition}
\newtheorem{example}[theorem]{Example}

\Crefname{fact}{Fact}{Facts}
\Crefname{assumption}{Assumption}{Assumptions}

\title{M-estimation with e-statistics}

\author[1]{Hongjian Wang}
\author[2]{Aaditya Ramdas}
\affil[1, 2]{Department of Statistics and Data Science, Carnegie Mellon University}
\affil[2]{Machine Learning Department, Carnegie Mellon University} 
\affil[ ]{\texttt{ \{hjnwang,aramdas\}@cmu.edu  }}
\date{\today}

\begin{document}

\maketitle

\begin{abstract}
   We present a theory of point estimation with e-statistics (e-values and e-processes) by introducing the ``ME-estimator'': the parameter that minimizes the corresponding e-statistic, or the evidence against it.
   Our approach is based on the intuitive idea of e-statistics as a measure of evidence and betting pay-off, and naturally generalizes the classical method of maximum likelihood estimation.
 First, we establish the consistency as well as the almost sure convergence rate for ME-estimators relating to the high-probability bounds on the size of the confidence set derived from thresholding the e-statistics, an approach that sets ME-estimators apart from traditional M-estimators. Second, we conduct classical M-estimator-style analysis on the consistency and asymptotic normality of ME-estimators in the bounded mean estimation setting, discussing the notion of efficiency (or lack thereof) from various choices of betting strategy. Our work brings e-statistics, a fundamental tool for inference and uncertainty quantification, to the space of estimation.
\end{abstract}


\section{Introduction}

E-values and e-processes, which we collectively refer to as ``e-statistics'', have become the cornerstone of modern safe, anytime-valid statistical inference \citep{ramdas2023game}. An e-value is a nonnegative statistic that has expected value $\le 1$ under any null distribution, and an e-process is formed by a sequence of statistics that yields an e-value at any stopping time. Recent literature has demonstrated that e-statistics are central in both fixed-time and sequential testing and uncertainty quantification. In this paper, we explore an aspect of e-statistics that has barely been visited by the past literature: can we form \emph{point estimators} from e-statistics that enjoy both conceptual clarity and theoretical guarantees? We answer the question positively by formulating a form of M-estimators with e-statistics.

We set the stage with a parametric statistical model.
Let $\Theta$ be a space of statistical parameters. For each $\theta \in \Theta$, there is a set of data-generating distributions $\cP_{\theta}$. Probabilistic statements that hold true under all $P \in \cP_\theta$ are denoted with symbols $\Pr_{\theta}$ and $\Exp_\theta$. There is also an unknown ``ground truth'' distribution, denoted by $P^*$ and belonging to $\cP_{\theta^*}$, which generates an infinite stream of i.i.d.\ observations $(X_n) = X_1, X_2 , X_3 \dots$. Sometimes, we switch to notations $\Pr$ and $\Exp$ when stating generic probability statements independent of our statistical framework.

For each $\theta \in \Theta$ and sample size $n \ge 1$ suppose we can construct a statistic $M_n^\theta$ from $X_1,\dots, X_n \iid P^*$ that
\begin{itemize}
    \item is an \emph{e-value}, i.e.\ $M_n^\theta \ge 0$ and $\Exp_\theta M_n^\theta \le 1$; or, additionally,
    \item forms an \emph{e-process}, i.e.\ $\Exp_\theta M_\tau^\theta \le 1$ for any stopping time $\tau$ on the data filtration.
\end{itemize}
We study the following estimator for the ground truth $\theta^*$ in this paper:
\begin{definition}[ME-estimator] \label[definition]{def:me}
    We let $\tilde\theta_n^{(E)}$ be any minimizer of e-statistics:
    \begin{equation}\label{eqn:e-estimator}
        \tilde \theta_n^{(E)} \in \argmin_{\theta \in \Theta} M_n^\theta,
    \end{equation}
    and call it a \emph{ME-estimator} (pronounced as ``mee'') derived from the family of e-statistics $M_n^\theta$. If it happens to be unique, we will call it \emph{the} ME-estimator.
\end{definition}
The name ``ME'' stands for ``minimizing e-statistic'' or ``minimizing evidence''.
Here, the $\argmin$ is a nonempty set (with possibly multiple elements, tie broken arbitrarily), often due to the compactness of $\Theta$ and the continuity $\theta \mapsto M_n^\theta$. When the minimum uniquely exists, we shall use ``$=$'' instead of ``$\in$'' in \eqref{eqn:e-estimator}. ME-estimators are related to the concept of the ``e-posterior'' studied by \cite{grunwald2023posterior}, where the ``e-posterior minimax rule'' with the negative indicator loss function $L(\theta, \theta') = -\id_{ \{ \theta = \theta'\} }$ defines an action equivalent to our ME-estimator \eqref{eqn:e-estimator}, for settings in which the e-posterior is well-behaved.

One may also recall M-estimators in the classical statistics literature. However, due to the e-statistics, ME-estimators \eqref{eqn:e-estimator} sometimes enjoy additional properties and analysis shortcuts unseen in the M-estimation literature, and sometimes admit the classical M-estimator analysis in a way that reflects the power of the e-statistics. ME-estimators are not immediately subsets of M-estimators considered in the literature, as $M^\theta_n$, even when applied to i.i.d.\ data, often contains unconventional martingale dependence structure introduced by the sequential testing strategy. We discuss the detailed comparison to M-estimators in \cref{sec:comp-to-m-estimators}.

\paragraph{Contributions} 
The main contributions of this paper are the demonstration of
these fundamental properties of $ \tilde \theta_n^{(E)}$:
\begin{itemize}
    \item Conceptual validity (using $\tilde \theta_n^{(E)}$ to estimate $\theta^*$ agrees with our philosophical understanding of e-statistics),
    \item Consistency (convergence $\tilde\theta_n^{(E)} \to \theta^*$ in probability or $P^*$-almost surely),
    \item Rate of convergence ($\mathcal{O}_{a.s.}$ and $o_{a.s.}$ rate of $|\tilde\theta_n^{(E)} - \theta^*|$ under $P^*$), and
    \item Asymptotic normality (convergence of $\sqrt{n}(\tilde\theta_n^{(E)} - \theta^*)$ to a normal distribution),
\end{itemize}
given that the e-statistics $\{M_n^\theta : \theta \in \Theta \}$ are ``powerful enough''. 

A recurring instantiation of the aforementioned e-statistics and ME-estimator paradigm that we consider in this paper is the following \emph{bounded betting} setup, studied by various authors including \cite{shafer2005probability,shafer2019game,shafer2021testing,waudby2020estimating,orabona2023tight,voracek2025starbets,wang2026almost} and is widely recognized as one of the most fundamental e-statistics constructions.
\begin{example}[Bounded betting] \label[example]{ex:betting}
    Let $\Theta = [0,1]$ and $\cP_\theta$ be the set of all $[0,1]$-bounded distributions with mean $\theta$. The bounded betting e-process (wealth process, test martingale) is
\begin{equation}\label{eqn:betting-debut}
    W_n^\theta = \prod_{t=1}^n (1 + \lambda_t(\theta) (X_t - \theta) ),
\end{equation}
where, for each $\theta$, the \emph{bet fraction} $\lambda_t(\theta) \in[- \frac{1}{1-\theta}, \frac{1}{\theta} ]$ is a predictable process depending continuously on $\theta$. Additionally, we may consider a convex combination (mixture) of these processes arising from different choices of bet fraction. For example,
\begin{equation}
    W_n^\theta = \frac{1}{2} \prod_{t=1}^n (1 + \lambda_t(\theta) (X_t - \theta) ) + \frac{1}{2} \prod_{t=1}^n (1 - \lambda_t(\theta) (X_t - \theta) )
\end{equation}
and
\begin{equation}
    W_n^\theta = \int_{-1}^1 \prod_{t=1}^n (1 + \lambda (X_t - \theta) ) \pi(\d  \lambda)
\end{equation}
where $\pi$ is a probability measure on $[-1, 1]$. We refer to a specific realization of $(W_n^\theta)$ (i.e.\ choice of $\lambda_t(\theta)$, $\pi$, etc.) as a ``strategy''.
With any strategy, a ME-estimator
\begin{equation}
    \tilde \theta_n^{(E)} \in \argmin_{\theta \in [0,1]} W_n^{\theta}
\end{equation}
exists due to the continuity of $\theta \mapsto W_n^\theta$ and the compactness of $[0,1]$.
\end{example}

While bounded betting is our primary focus of application, our techniques and results generalize well beyond that. For example, we shall later mention the heavy-tailed unbounded mean (\cref{ex:catoni-cs-concentration}), the PAC-Bayes multivariate mean (\cref{ex:pacbayes-sn}), and the quantile testing (\cref{ex:kt-quantile}) e-processes as other instantiations of our methods.


\subsection{Evidence- and game-theoretic justification of ME-estimators}\label{sec:why-e-est}

The definition of e-values and e-processes has led to \emph{inference} (hypothesis testing with type-I error guarantee) and \emph{uncertainty quantification} (confidence intervals and confidence sequences with coverage guarantee) for the statistical parameter. Indeed, if $M_n^\theta$ is an e-value under $\theta$, one routinely rejects the null $\theta^* = \theta$ if $M_n^\theta \ge 1/\alpha$. By Markov's or Ville's inequality, this rejection rule guarantees a type 1 error rate bounded by $\alpha$.

The duality between testing and confidence sets then leads to the following standard construction of
a $(1-\alpha)$-confidence set for the ground truth $\theta^*$: 
\begin{equation}\label{eqn:eci}
    \CI_n^{\alpha} := \{ \theta: M_n^\theta \le 1/\alpha \},
\end{equation}
that is, the collection of parameters whose e-statistics are less than $1/\alpha$. This is usually referred to as an ``e-CI''. A ME-estimator defined in \cref{def:me} is a clear variation of theme on this confidence set. In fact, as we shall reveal soon, analyzing ME-estimators can sometimes be reduced to analyzing e-CIs.

Past authors in the field of e-statistics have very occasionally mentioned the idea of ME-estimators. For example, \citet[Figure 18]{waudby2020estimating} plot the map $\theta \mapsto 1/M_n^{\theta}$ in several examples and advise casually that its maximum is ``a reasonable point estimator". 
In general, however, there has been a lack of effort in understanding both the ME-estimators \eqref{eqn:e-estimator} and the more general question ``how to do point estimation with e-statistics". Behind this lack of a point estimation theory of e-statistics lies an important distinction between traditional statistical inference (e.g.\ parametric maximum likelihood estimation) and e-statistics-driven statistical inference. In traditional statistical inference, a point estimator is often obtained first via some score/loss criterion as the solution to the corresponding optimization problem, and tests and confidence sets follow subsequently from the properties of the point estimator. With e-statistics, a valid test against the null hypothesis now takes on a primitive role as opposed to being derivative of some estimator. It is our focus in this paper to discuss the derivative point estimator in this case. 

We propose ME-estimators \eqref{eqn:e-estimator} on the grounds of two central interpretations of e-statistics. 

\begin{itemize}
    \item First, e-statistics are commonly interpreted as a theory of \emph{evidence}.
$M_n^\theta$, the e-value testing the null $\theta^* = \theta$, measures the amount of evidence \emph{against} this null carried by the data. Indeed, if the unll $\theta^* = \theta$ was true, it would be unlikely for $M_n^\theta$ to be large due to Markov's inequality: $\Pr_{\theta}(M_n^\theta \ge x) \le x^{-1}$. Therefore the point estimator \eqref{eqn:e-estimator} has the following meaning: it
     is the parameter candidate \emph{against which observed data carries the least amount of evidence}. 
\item Second, e-statistics are the fundamental quantities in game-theoretic statistics, where they are seen as \emph{betting payoffs} in a game of ``testing by betting'' \citep{shafer2021testing}. The game involves two parties, a seller and a buyer. The risk-neutral seller believes that the null $\theta^* = \theta$ holds true, and thus offers a bet priced at $1$ unit pre-observation that pays $M_n^\theta$ units post-observation. If the null $\theta^* = \theta$ is indeed true, then these bets are in the seller's favor since $\Exp_\theta M_n^\theta \le 1$. The opportunistic buyer, on the other hand, challenges the null $\theta^* = \theta$ and believes in reality $M_n^\theta$ will likely turn out large under $P^*$, thus investing in the bet. A ME-estimator \eqref{eqn:e-estimator} therefore has the following optimization interpretation from the seller's perspective: since $1 - M_n^\theta$ is seller's realized profit when offering the bet on $\theta^* = \theta$,
    minimizing $M_n^\theta$ corresponds to picking the \emph{empirically} (i.e.\ hind-sight) most profitable bet to offer.

\end{itemize}

To summarize, on a philosophical level, a ME-estimator $\tilde \theta_n^{(E)}$ makes both evidence- and game-theoretic sense. In fact, we will revisit the second perspective and consider the buyer as well as we develop our ideas.
The remaining task of this paper is to further justify the conceptually valid ME-estimators by showing the quantitative frequentist properties they satisfy.

\subsection{Comparison to classical M-estimators}\label{sec:comp-to-m-estimators}

From the definition of $\tilde\theta_n^{(E)}$ in \eqref{eqn:e-estimator}, it is tempting to think that $\tilde\theta_n^{(E)}$ is a special kind of M-estimator as it is defined as the minimizer of a random function. The theoretical statistics literature has developed numerous standard methods (see e.g.\ Chapter 5 in \cite{van2000asymptotic}) for studying the consistency and normality of M-estimators. Do these standard methods apply? We provide a two-fold answer to this question, which sets the agenda for the rest of this paper.

    \begin{itemize}
        \item Classical M-estimator theory typically concerns
        \begin{equation}
            \argmin_{\theta} \sum_{t=1}^n f(\theta, X_t)
        \end{equation}
        that converges to the unique minimizer of $\theta \mapsto \Exp f(\theta, X)$ for some \emph{fixed} map $f$. Powerful (log-)e-statistics usually don't take the form of $\sum_{t=1}^n f(\theta, X_t)$. Rather, they are usually
        \begin{equation}
             \argmin_{\theta} \sum_{t=1}^n f_t(\theta, X_t)
        \end{equation}
        for a predictable sequence of maps $f_t$. See e.g.\ the betting example \eqref{eqn:betting-debut}. This introduces a major structural difference between M- and ME-estimators.
        However, in certain cases, the map sequence $(f_t)$ ``stabilizes'' over time, and we might be able to carry out the standard M-estimator analysis. We investigate these cases later, in \cref{sec:classical}. In fact, as we shall see, the successful adoption of the standard M-estimator analysis depends on the \emph{e-power} of the underlying e-statistics, which is the fundamental power metric for e-statistics.
        
        \item As it turns out, the convergence analysis of ME-estimators can take much advantage of the fact that they are derived from e-statistics. The safety guarantee of e-statistics, due to Markov's and Ville's inequalities, provides a natural concentration set for ME-estimators and leads to almost sure, sometimes fast convergence. We discuss this approach first, in \cref{sec:growth}.
    \end{itemize}

\section{Convergence analysis via growth of e-statistics}\label{sec:growth}

In this section, we analyze the convergence of ME-estimators loosely based on the following idea. First, $\tilde\theta_n^{(E)} \in \CI_n^\alpha$ where $\CI_n^\alpha$ is the confidence set formed by thresholding e-statistics: $\{ \theta: M_n^{\theta} \le 1/\alpha \}$ (when it is not empty -- a caveat we address later). Second, $\theta^* \in \CI_n^\alpha$ with probability at least $1-\alpha$. Therefore, the size of $\CI_n^\alpha$ upper bounds the distance between $\tilde\theta_n^{(E)}$ and $\theta^*$. We show that this idea ultimately leads to the almost sure convergence of $\tilde\theta_n^{(E)} \to \theta^*$ with quantitative rates, beginning with a slight tweak of the coverage property of these ``confidence sets'' below.

\subsection{Almost sure confidence sets}\label{sec:doubling}

Traditional confidence intervals and confidence sequences track the unknown parameter with high probability at some or all fixed sample size(s). \cite{Naaman2016as} introduces a variant of the concept where the sequence of sets tracks the unknown parameter with probability 1, but from an indefinite sample size onward. We slightly reformulate some of the discussions found in the work of \cite{Naaman2016as}.

\begin{definition}
    Let $( I_n )_{n \ge 1}$ be a sequence of random sets and $x$ a non-random element. We say $( I_n )$ is a sequence of \emph{almost sure confidence sets} for $x$ if $x$ belongs to all but finitely many of them almost surely; that is,
    \begin{equation}
        \Pr\left( \text{there exists }N>1\text{ such that } x \in \bigcap_{n = N}^\infty I_n \right) = 1.
    \end{equation}
\end{definition}

A sequence of traditional confidence intervals, with properly chosen error levels, is almost sure by the Borel-Cantelli lemma. We refer to this procedure as ``almost sure-ifying'' the confidence intervals.

\begin{proposition}[Almost sure-ified confidence intervals] \label[proposition]{prop:almost-sureification}
    For each $n$, suppose $\CI_n^\alpha$ is a $(1-\alpha)$-confidence interval for $\theta$. Then, if the sequence $( \alpha_n )$ satisfies $\sum_{n=1}^\infty \alpha_n < \infty$, the set sequence $( \CI_n^{\alpha_n} )$ is almost sure for $\theta$.
\end{proposition}
\begin{proof}
    Since $ \sum_{n=1}^\infty \Pr(\theta \notin \CI_n) < \infty$, by the Borel-Cantelli lemma,
    \begin{equation}
        \Pr( \theta \notin \CI_n \text{ for infinitely many }n   ) = 0
    \end{equation}
    concluding the proof.
\end{proof}

While \cite{Naaman2016as} introduces almost sure confidence sets mainly as an alternative to the traditional fixed-$\alpha$ confidence sets to resolve Lindley's paradox in testing, we note that these sets also lead to consistent point estimators. Specifically,
a \emph{shrinking} almost sure sequence of confidence sets leads to a consistent sequence of point estimators, by choosing \emph{any} point in the confidence set (e.g.\ the upper/lower boundary and center in one dimension). More generally, an almost sure shrinkage rate of these almost sure confidence sets also implies the same almost sure convergence rate of these estimators.

\begin{proposition}\label[proposition]{prop:as-conv}
    Let $\{ I_n  \}$ be an almost sure sequence of confidence sets for $\theta^*$. For any estimator sequence $(\tilde \theta_n)$ such that 
    \begin{equation}
        P^*( \text{$\tilde \theta_n \in I_n$ for all but finitely many $n$} ) = 1
    \end{equation}
then
    \begin{equation}
         P^*(\tilde\theta_n \to \theta^* \, | \diam(I_n) \to 0 )=1.
    \end{equation}
    More generally,
    \begin{gather}
        P^*( |\tilde\theta_n - \theta^*| =  \mathcal{O}(r_n)  \, | \diam(I_n) = \mathcal{O}(r_n) )=1, \label{eqn:asrate-1} \\ P^*( |\tilde\theta_n - \theta^*| =  o(r_n)  \, | \diam(I_n) = o(r_n) )=1. \label{eqn:asrate-2}
\end{gather}
\end{proposition}
\begin{proof} 
Let $E$ be the event $ \left\{ \text{there exists }N>1\text{ such that } \{\theta^*, \tilde \theta_n \} \subseteq  I_n  \text{ for all }n \ge N \right\}$. Then $P^*(E) = 1$.
For each $\omega \in E$, let $N(\omega)$ be such $N$.
   When $n \ge N(\omega)$,
   $|\tilde\theta_n - \theta^* | \le  \diam(I_n)$. This shows that
   \begin{align}
 E \cap  \{  \diam(I_n) \to 0 \}  &\subseteq  \{ \tilde\theta_n \to \theta^* \}, \\
 E \cap  \{  \diam(I_n) = \mathcal{O}(r_n) \}   &\subseteq  \{ |\tilde\theta_n - \theta^*| =  \mathcal{O}(r_n) \}, \\
  E \cap  \{  \diam(I_n) = o(r_n) \}   &\subseteq  \{ |\tilde\theta_n - \theta^*| =  o(r_n) \}.
   \end{align}
   concluding the proof.
\end{proof}

\subsection{Transfer theorems for ME-estimators}

Let us now apply \cref{prop:as-conv} to ME-estimators.
Recall that $(M^\theta_n)$ are e-values or form an e-process under $\Pr_{\theta}$ for $\theta \in \Theta$. For any $\alpha \in (0, 1/2]$, recall from \cref{sec:why-e-est} that we denote the threshold set by
\begin{equation}
        \CI_n^{\alpha} = \{ \theta: M^\theta_n \le 1/\alpha \},
\end{equation}
which is a $(1-\alpha)$-confidence set for $\theta^*$; and forms a $(1-\alpha)$-confidence sequence for $\theta^*$ if $(M^\theta_n)$ are e-processes, in which case we write $\CS_n^\alpha$ instead of $\CI_n^\alpha$. These confidence sets are a common application of e-statistics.

Our first observation is that, since $\CI_n^\alpha$ contains $\theta^*$ with probability at least $1-\alpha$, it is nonempty with at least the same probability; as long as it is nonempty, it contains ME-estimators due to their respective definitions. Formally, for any $\alpha \in (0,1)$ and any ME-estimator $ \tilde\theta_n^{(E)}$:
\begin{equation}\label{eqn:me-estimator-in-CI}
    P^*( \tilde\theta_n^{(E)} \in \CI_n^\alpha )= P^*(M_n^{\tilde \theta_n^{(E)}} \le 1/\alpha) \ge P^*( M_n^{\theta^*} \le 1/\alpha )  \ge 1-\alpha.
\end{equation}
The coverage \eqref{eqn:me-estimator-in-CI} is analogous to the parameter coverage $P^*( \theta^* \in \CI_n^\alpha ) \ge 1-\alpha$, and we can apply the same almost sure-ification procedure \cref{prop:almost-sureification}. This yields the following instantiation of \cref{prop:as-conv}.
\begin{theorem}\label{thm:BC-ev}
    Let $(\alpha_n)$ be a sequence of numbers in $(0,1)$ such that $\sum_{n=1}^\infty \alpha_n < \infty$. Then, the sequence of confidence sets defined by
\begin{equation}
    \CI_n^* = \CI_{n}^{\alpha_n}= \{ \theta: M^\theta_n \le 1/\alpha_n \}
\end{equation}
is almost sure for $\theta^*$, and also contains ME-estimators $\tilde \theta_n^{(E)} \in \CI_n^*$ for all but finitely many $n$ almost surely. Consequently, \eqref{eqn:asrate-1} and \eqref{eqn:asrate-2} hold for $I_n = \CI_n^*$, $\tilde\theta_n = \tilde \theta_n^{(E)}$, and any rate sequence $(r_n)$.
\end{theorem}

The proof of \cref{thm:BC-ev} is a simple utilization of the Borel-Cantelli lemma on \eqref{eqn:me-estimator-in-CI}, as we did for \cref{prop:almost-sureification}, so we omit it from the paper.

Additionally, it is worth noting that we can weaken the Borel-Cantelli condition $\sum_n \alpha_n < \infty$ in \cref{thm:BC-ev} if we obtain ME-estimators from
 a family of e-processes $(M_n^\theta)$. We prove \cref{thm:BC-epr} below in \cref{sec:pf-bc-epr}. 


\begin{theorem}\label{thm:BC-epr} Suppose $(M_n^\theta)$ are e-processes.
    Let $1 = n_1 < n_2 < \dots$ be a sequence of integers, and $\alpha_1 > \alpha_2 > \dots$ be a vanishing sequence of positive numbers. Denote by $t(n)$ the integer $t$ that satisfies $n \in [ n_{t}, n_{t+1} )$. Then, the sequence of confidence sets defined by
\begin{equation}
    \CS_n^* = \CS_{n}^{\alpha_t} = \{ \theta: M^\theta_n \le 1/\alpha_{t(n)} \}.
\end{equation}
is almost sure for $\theta$. Consequently, \eqref{eqn:asrate-1} and \eqref{eqn:asrate-2}
     hold for $I_n = \CS_n^*$, and also contains ME-estimators $\tilde \theta_n^{(E)} \in \CS_n^*$ for all but finitely many $n$ almost surely. Consequently, \eqref{eqn:asrate-1} and \eqref{eqn:asrate-2} hold for $I_n = \CS_n^*$, $\tilde\theta_n = \tilde \theta_n^{(E)}$, and any rate sequence $(r_n)$.
\end{theorem}


The theorems above show that the strong consistency and convergence of ME-estimators are deeply connected to the following problem in the e-statistics literature: \emph{given the e-statistic family $\{ M_n^{\theta} : \theta \in \Theta \}$, how can one quantify the size of the corresponding confidence set $\diam(\CI_n^\alpha)$?} In particular, we would like the almost sure rate of shrinkage when $\alpha = \alpha_n$ is plugged in at $\diam(\CI_n^{\alpha})$: $P^*( \diam(\CI_n^{*} ) = \mathcal O(r_n) ) = 1$ or $P^*( \diam(\CI_n^{*} ) = o(r_n) ) = 1$. We study this in the following subsection.

\subsection{Bounding $\diam(\CI_n^*)$ via uniform growth inequalities}

A confidence set $\CI_n^\alpha$ obtained by thresholding e-statistics usually has random diameters. To invoke \cref{thm:BC-epr}, we need $\mathcal{O}_{a.s.}$ or $o_{a.s.}$ shrinkage rates on $\diam(\CI_n^*)$. This, we note, may be obtained once again with the technique of almost sure-ification in \cref{prop:almost-sureification}.
Let us orient the reader with an existing result in the literature where a nonasymptotic analysis exists on $\diam(\CI_n^\alpha)$, then explore the generality.

Our example comes from \cite{Wang2023Catoni}, where an exponentially decaying confidence sequence is constructed for heavy-tailed (finite variance) data via a pair of supermartingales.

\begin{example}[Catoni-style CS] \label[example]{ex:catoni-cs-concentration}

\citet[Corrolary 10.1]{Wang2023Catoni} show that the diameter $\diam(\CS_n^{\alpha})$ of the ``Catoni-style confidence sequence'' enjoys the following (informal) concentration bound
\begin{equation}\label{eqn:catoni-cs-bound-informal}
     P^*\left(\diam(\CS_n^{\alpha}) \le  {\mathcal{O}}(n^{-1/2}
    \log n(\log(1/\varepsilon) + \log(1/\alpha)) ) \right) \ge 1 - \varepsilon.
\end{equation}
As we can almost sure-ify a sequence of confidence sets or a confidence sequence by a vanishing sequence of \emph{coverage} error levels $\alpha$, we can also almost sure-ify this high-probability bound on the size $\diam(\CS_n^\alpha)$ by a vanishing sequence of \emph{concentration} error levels $\varepsilon$. To this end, let us take
$\alpha_n = \varepsilon_n = 1/2n^2$ in \eqref{eqn:catoni-cs-bound-informal}. Then, $(\CS_n^{\alpha_n})$ is an almost sure sequence of confidence sets, and by the Borel-Cantelli lemma,
\begin{equation}
  \diam(\CS_n^{\alpha_n}) = \mathcal{O}_{a.s.}(\log n /\sqrt{n}),
\end{equation}
which is also an almost sure convergence rate of ME-estimators in this case.
\end{example}

The example above involves a random-size confidence set, whose high-probability bound is proved bespokely by \cite{Wang2023Catoni} via a supermartingale argument which does not easily generalize to other e-statistics. Therefore, we establish a widely applicable ``master theorem''
that bridges the \emph{uniform} growth of e-statistics, the high-probability bound of the confidence set, and consequently the almost sure convergence of ME-estimators.

\begin{theorem}[Master theorem of uniform growth]\label{thm:master} Let $(M_n^\theta)$ be the e-values or e-process testing the null $\theta$. Fixing a data-generating distribution $P^* \in \cP_{\theta^*}$,
    suppose there exists a data-dependent function (which may depend on \emph{both} $P^*$ \emph{and} the data $X_1,\dots, X_n$) $D_n: \mathbb R^{\ge 0}\times [0,1) \to \mathbb R^{\ge 0}$ that are strictly increasing in its first argument, such that the following uniform concentration holds
    \begin{equation}\label{eqn:gn-conc}
  P^* \left\{ \forall \theta \in \Theta, \ \log M_n^{\theta} \ge D_n( |\theta - \theta^* |, \delta)  \right\} \ge 1-\delta, \quad \forall \delta \in [0,1)
\end{equation}
for all $n \ge N_0$.
Then, denoting the inverse of $D_n( \cdot, \delta)$ by $\varphi_n(\cdot, \delta)$, the size of the inverted confidence set $\CI_n^\alpha = \{ \theta : \log M_n^{\theta} \le \log(1/\alpha)  \}$ is concentrated as
\begin{equation}\label{eqn:hp-bound-CI}
    P^* \left\{  \diam(\CI_n^\alpha) \le 2 \varphi_n (\log(1/\alpha) ,\delta) \right\} \ge 1-\delta,
\end{equation}
and the error of any ME-estimator is concentrated as
\begin{equation}\label{eqn:hp-bound-err}
    P^* \left\{  | \tilde\theta_n^{(E)} - \theta^* | \le  \varphi_n (\log(1/\alpha) ,\delta) \right\} \ge 1-\delta -\alpha.
\end{equation}
From either, we may deduce that ME-estimators have a rate of convergence of
\begin{equation}\label{eqn:as-rate-me}
  \tilde\theta_n^{(E)} = \theta^* + \mathcal{O}_{a.s.}(\varphi_n((1+\varepsilon)\log n ,n^{-1-\varepsilon}) ),
\end{equation}
where $\varepsilon > 0$ is arbitrary.
\end{theorem}

The monotonicity assumption on $D_n(\cdot, \delta)$ and the existence of an inverse $\varphi_n(\cdot, \delta)$ can be relaxed to
\begin{equation}
  \forall h > 0,  D_n(\varphi_n(  h, \delta ) ,\delta) = h, \text{and }  D_n(v, \delta)  > h \text{ for all } v > \varphi_n(  h, \delta ),
\end{equation}
which holds for e.g.\ convex $D_n(\cdot, \delta)$.
The straightforward proof of this theorem involves applying the triangle inequality for \eqref{eqn:hp-bound-CI} and applying \eqref{eqn:me-estimator-in-CI} for \eqref{eqn:hp-bound-err}, which 
we omit from the paper. We refer to an instance of \eqref{eqn:gn-conc} as a ``uniform growth inequality'' (UGI).
The concepts involved in the theorem can be presented by the following diagram, where arrows marked with ``B-C'' are applications of the Borel-Cantelli lemma.

$$
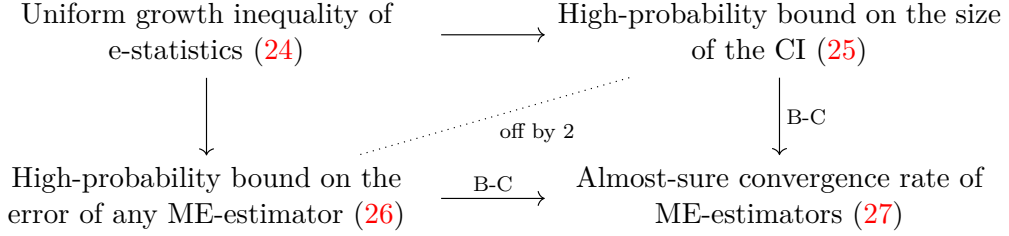
\begin{tikzcd}[column sep=large, row sep=large]
    \parbox{0.4\textwidth}{\centering Uniform growth inequality of e-statistics \eqref{eqn:gn-conc}} \arrow[r] \arrow[d] & \parbox{0.4\textwidth}{\centering High-probability bound on the size of the CI \eqref{eqn:hp-bound-CI}} \arrow[d,"\text{B-C}"] \arrow[dl,dash,dotted,"\text{off by 2}"] \\
    \parbox{0.4\textwidth}{\centering High-probability bound on the error of any ME-estimator \eqref{eqn:hp-bound-err}} \arrow[r,"\text{B-C}"] &   \parbox{0.4\textwidth}{\centering Almost-sure convergence rate of ME-estimators \eqref{eqn:as-rate-me}}
\end{tikzcd}$$
   
 The functions $(D_n)$ in the UGI \eqref{eqn:gn-conc} characterize the growth trend of the log-e-statistics (by time and effect size), which can be deterministic or sample-dependent. The inverses, $(\varphi_n)$, can consequently also be sample-dependent, but their almost sure shrinkage rate is often easy to determine, so that via \eqref{eqn:as-rate-me}:
 \begin{equation}
     \varphi_n((1+\varepsilon)\log n ,n^{-1-\varepsilon}) = \mathcal{O}_{a.s.}(r_n) \implies  \tilde\theta_n^{(E)} = \theta^* + \mathcal{O}_{a.s.}(r_n).
 \end{equation}
 The key idea to obtain a UGI is the \emph{decoupling} of $|\theta - \theta^*|$ from the log-e-statistics, by means of deterministic or concentration inequalities. 

One may also expand the UGI \eqref{eqn:gn-conc} as
\begin{equation}
    \log M_n^{\theta} \ge D_n( |\theta - \theta^*| , 0 ) - ( D_n( |\theta - \theta^*| , 0 )  - D_n( |\theta - \theta^*| , \delta ) ),
\end{equation}
consisting of a mean term and a downside fluctuation term.
This likens the UGI to regret bounds proved in the context of online learning. Indeed, many log-e-processes do satisfy regret bounds (see e.g.\ \cite{orabona2023tight}). In fact, we shall discuss regret bounds in the bounded betting context later. However, two crucial distinctions exist between UGI and regret bounds. First, UGI satisfies with high probability, whereas regret bounds usually hold on all sample paths (a notable exception is the recent work by \cite{agrawal2025eventually} on high-probability regret bounds). Second, more importantly, UGI and regret bounds compare the realized log-e-statistic $\log M_n^\theta$ to different benchmarks: for UGI, it is the ``mean'' driving forces $D_n(|\theta - \theta^*|, 0)$ that the actual log-e-statistics fluctuate around \emph{in unison across the $\theta$ dimension}; whereas for regret bounds, it is hind-sight optimal and there is no notion of uniformity over $\theta$. The meaning of UGI will become much clearer when we illustrate it with examples.

\subsection{Some examples of uniform growth inequalities}\label{sec:betting-growths}

We investigate a few e-process-based tests where \cref{thm:master} is satisfied. We begin with the bounded betting, \cref{ex:betting}. The ground truth is $(X_n) \iid P^*$ on $[0,1]$ with mean $\theta^*$ and variance $\sigma^2$. We instantiate \cref{thm:master} on the following simple strategy, where we show an $n^{-1/2}\log n$ almost sure rate of ME-estimators.

\begin{example}[Deterministic hedging] \label[example]{ex:det-hedge}
Consider a ``hedging'' strategy that places decreasing bets on both ``heads'' and ``tails'' symmetrically and deterministically. The wealth process reads
\begin{equation}
W_n^\theta =  \frac{1}{2} \prod_{t=1}^n (1 + \lambda_t(X_t - \theta) ) + \frac{1}{2} \prod_{t=1}^n (1 - \lambda_t(X_t - \theta) ).
\end{equation}
We set $\lambda_t = 1/\sqrt{3t}$ for demonstration but other nonrandom $\asymp t^{-1/2}$ choice would all maintain the argument.
With this setup, we can show that \cref{thm:master} holds with the uniform growth inequality:
\begin{equation}\label{eqn:hedging-ugi}
  P^* \left( \forall \theta \in [0,1],\  \log W_n^\theta \ge 2\sqrt{n/3}|\theta - \theta^*| - \log 2 - \frac{\log n+1}{3}   - \sqrt{\frac{\log(2en/\delta)}{6}} \right) \ge 1 - \delta.
\end{equation}
Consequently, the inverted confidence set $\CS_n^\alpha$ satisfies
\begin{equation}
    P^* \left( \diam(\CS_n^{\alpha}) \le  \frac{\log(1/\alpha) + \log 2 + \frac{\log n + 1}{3} + \sqrt{\frac{\log(2en/\delta)}{6}}  }{\sqrt{n/3}}  \right) \ge 1-\delta,
\end{equation}
and therefore
\begin{equation}
   \tilde \theta_n^{(E)} = \theta^* + \mathcal{O}_{a.s.}(n^{-1/2}\log n ).
\end{equation}
\end{example}

We prove the UGI \eqref{eqn:hedging-ugi} of this betting strategy in \cref{sec:pf-det-hedge}. It is worth remarking that in this UGI, the ``growth rate'' $\log W_n^{\theta} \asymp \sqrt{n}$ is sublinear with respect to the sample size $n$, which is arguably the actual growth rate due to the $\lambda_n \asymp n^{-1/2}$ vanishing bet rate of the strategy. There exist many other betting strategies that can achieve a linear growth rate $\log W_n^{\theta} \asymp n$, one of which we present immediately next. However, despite not being power-optimal, the hedging strategy nonetheless produces ME-estimators of near-optimal rate. This example shows that a family of e-statistics with suboptimal (in $n$) UGI can lead to a near-optimal convergence rate of ME-estimators, using our method presented in this section.

Now, we present another betting strategy whose UGI growth rate is the faster-growing $\log W_n^{\theta} \asymp n$, while having the same $\mathcal{O}_{a.s.}(n^{-1/2}\log n )$ rate.

\begin{example}[KT] \label[example]{ex:kt-ugi} Consider the Krichevsky-Trofimov-type (KT) strategy where the fraction $\lambda_t$ is an appropriately shifted sample mean of past observations \citep{orabona2016coin}:
\[ \lambda_t^{KT}(\theta) = \frac{1/2 + \sum_{i=1}^{t-1} (X_i - \theta)}{t c }. \]
The wealth process is $W_n^\theta = \prod_{t=1}^n (1 + \lambda_t^{KT}(\theta) (X_t - \theta))$. We set $c=2$ below, but the argument below extends to general $c \ge c_{\min}\approx 1.46$ such that $\log (1+x) \ge x-x^2$ holds for all $|x|\le 1/c$.
We can show that \cref{thm:master} holds with the uniform growth inequality:
\begin{equation}\label{eqn:kt-ugi}
P^* \left( \forall \theta \in [0,1], \ \log W_n^\theta \ge \frac{n}{4} |\theta - \theta^*|^2 - \sqrt{n}K |\theta - \theta^*| - 2\sqrt{C_n}K - C_n \right) \ge 1 - \delta,
\end{equation}
where $K = \sqrt{\frac{1}{8} \log \frac{6}{\delta}}$ 
and $C_n = \frac{\pi^2}{12} + \frac{1}{8} \left(\log \frac{6n}{\delta}\right) (\log n + 1)$.
Consequently, the inverted confidence set $\CS_n^\alpha$ satisfies
\begin{equation}
    P^* \left( \diam(\CS_n^{\alpha}) \le \frac{4 K}{\sqrt{n}} + 2\sqrt{\frac{2\sqrt{C_n}K + C_n  + \log(1/\alpha)}{n/4}} \right) \ge 1-\delta,
\end{equation}
and therefore its ME-estimators achieve the near-optimal almost sure convergence rate:
\begin{equation}
   \tilde \theta_n^{(E)} = \theta^* + \mathcal{O}_{a.s.}(n^{-1/2}\log n ).
\end{equation}
\end{example}
The UGI \eqref{eqn:kt-ugi} is proved in \cref{sec:pf-kt}. The KT strategy, which later plays a central role in \cref{sec:classical}, has a fast, linear uniform growth rate \eqref{eqn:kt-ugi} that leads to the same $\mathcal{O}_{a.s.}(n^{-1/2}\log n )$ convergence rate of ME-estimators compared to the slower growing hedging strategy, \cref{ex:det-hedge}.
To intuitively understand this phenomenon, we remark that in general faster (in $n$) uniform growth of alternative e-processes does \emph{not} necessarily lead to smaller CSs and better ME-estimators. In particular, the fast shrinkage rate of CS and convergence rate of ME-estimators call for the growth rate to match $|\theta - \theta^*|$ \emph{in the vicinity of $\theta^*$}. A similar observation was made recently by \cite{agrawal2025eventually}.

We finally present a multivariate example.

\begin{example}[PAC-Bayes self-normalized e-process for bounded vectors] \label[example]{ex:pacbayes-sn} Let $\Sigma$ be a matrix upper bound of the common covariance matrix of $(X_n) \iid P^*$ on $ \mathbb R^d$. A log-e-process for the null $\Exp_\theta X_1 = \theta$ is (see Proof of Theorem 3.1 in \cite{chugg2023time})
\begin{equation}
  \log  M_n^\theta = \left\| \sum_{t=1}^n \lambda_t (X_t - \theta) \right\| - \sum_{t=1}^n \frac{\lambda_t^2}{6}\left(\|   X_t - \theta \|^2(1+\beta^{-1}) + 2\|\Sigma\|_{\mathrm{op}} + \beta^{-1} \tr\Sigma  \right) - \frac{\beta}{2},
\end{equation}
where $\beta > 0$ is an arbitrary tuning parameter, and let us (for demonstration) set $\lambda_t = 1/\sqrt{t}$.
With this setup, we can show that \cref{thm:master} holds with the uniform growth inequality:
\begin{multline}\label{eqn:pacbayes-ugi}
  P^* \Biggl( \forall \theta \in \mathbb R^d,\   \| \theta^* - \theta \|\sum_{t=1}^n \lambda_t  - \| \theta^* - \theta \|^2 \frac{1+\beta^{-1}}{3} \sum_{t=1}^n \lambda_t^2  \\
    -  \sum_{t=1}^n \frac{\lambda_t^2}{6}\left( 3\|   X_t - \theta^* \|^2(1+\beta^{-1}) + 4\|\Sigma\|_{\mathrm{op}} + 2\beta^{-1} \tr\Sigma  \right) - \beta - 2\log(1/\delta)  \Biggr) \ge 1 - \delta.
\end{multline}
Consequently, if we further assume that $\| X_1 \|\le L$, for large $n$ satisfying $\frac{\sum_{t=1}^n \lambda_t }{\sum_{t=1}^n \lambda_t^2} > \frac{2L(1+\beta^{-1})}{3}$, the inverted confidence set $\CS_n^\alpha$ satisfies
\begin{equation}\label{eqn:pacbayes-cs-bound}\scriptsize
    P^* \left( \diam(\CS_n^{\alpha}) \le  \frac{ 2\sum_{t=1}^n \frac{\lambda_t^2}{6}\left( 3\|   X_t - \theta^* \|^2(1+\beta^{-1}) + 4\|\Sigma\|_{\mathrm{op}} + 2\beta^{-1} \tr\Sigma  \right) + 2\beta + 4\log(1/\delta) + 2\log(1/\alpha)   }{
        \sum_{t=1}^n \lambda_t - \frac{2L(1+\beta^{-1})}{3} \sum_{t=1}^n \lambda_t^2
    }\right)   \ge 1-\delta,
\end{equation}
and therefore
\begin{equation}\label{eqn:pacbayes-as-rate}
   \tilde \theta_n^{(E)} = \theta^* + \mathcal{O}_{a.s.}(n^{-1/2}\log n ).
\end{equation}
\end{example}

We prove the UGI \eqref{eqn:pacbayes-ugi} and the almost sure rate \eqref{eqn:pacbayes-as-rate} in \cref{sec:pf-pacbayes-sn}. It is worth noting that the high-probability concentration bound, i.e.\ the fraction within \eqref{eqn:pacbayes-cs-bound}, is random, corresponding to our remark after \cref{thm:master} that $\phi_n$ can be sample-dependent.

\section{Classical M-estimator analysis of betting ME-estimators}\label{sec:classical}

We now discuss the consistency, asymptotic normality, and efficiency of ME-estimators in the bounded betting case (\cref{ex:betting}), using methods akin to those employed in the classical M-estimator literature (e.g.\ \cite{van2000asymptotic}). We keep adopting the setup in \cref{sec:betting-growths}, with $X_1, X_2, \dots \in [0,1]$ the sequence of observations with mean $\theta^*$ and variance $\sigma^2 > 0$.  The betting wealth reads, for $\theta \in [0,1]$,
\[ W_n^\theta = \prod_{t=1}^n \big(1 + \lambda_t(\theta) (X_t - \theta)\big), \]
where $\lambda_t(\theta)$ is a predictable betting fraction that depends on $\theta$, with ME-estimator
\[ \tilde{\theta}_n^{(E)} \in \argmin_{\theta} W_n^\theta  = \argmin_{\theta} \underbrace{\sum_{t=1}^n \log \big(1 + \lambda_t(\theta) (X_t - \theta)\big)}_{=: w_n(\theta)}. \]
Our analysis focuses on two predictable plug-in betting strategies. First, 
 the Krichevsky-Trofimov-type (KT) strategy which appeared earlier as \cref{ex:kt-ugi}, where the fraction $\lambda_t$ is an appropriately shifted sample mean of past observations
\[ \lambda_t^{KT}(\theta) = \frac{1/2 + \sum_{i=1}^{t-1} (X_i - \theta)}{t c }, \]
where $c \ge 1/0.6$ is some constant. Second, the aGRAPA (approximately growth rate adaptive to the particular
alternative) strategy introduced by \cite{waudby2020estimating}
$$\lambda_t^{aG}(\theta) = \frac{\bar{X}_{t-1} - \theta}{\hat{\sigma}_{t-1}^2 + (\bar{X}_{t-1} - \theta)^2}  \wedge 0.6 \vee -0.6, $$
where $\bar{X}_{t-1}$ is the empirical mean and $\hat{\sigma}_{t-1}^2$ is the empirical variance of past observations $X_1,\dots, X_{t-1}$.
The techniques presented here may extend to other betting strategies and potentially other e-statistics as well. Here, the clipping constant 0.6 is employed purely for mathematical convenience: the inequality $\log(1+x) \ge x- 0.9 x^2$ holds for all $x \ge -0.6$. Different clipping constants are also viable. We also remark that aGRAPA is the ``approximate'' version of a strategy called GRAPA, as discussed by \cite{waudby2020estimating}. We analyze aGRAPA instead of GRAPA in this paper because (1) GRAPA is much more computationally prohibitive than aGRAPA, but empirically they perform equally well; (2) the GRAPA bet fraction is defined by another M-estimation procedure, which makes it significantly harder to analyze the needed uniform convergence of bet fractions across $\theta$ dimension compared to the closed-form aGRAPA, the proof for which we later spell out as \cref{cor:consistency-kt-agrapa} in a few lines.

\subsection{Consistency}

Our analysis on the (weak) consistency of ME-estimators $\tilde \theta_n^{(E)}$ relies on the following key property of the betting strategy: the uniform convergence of $\frac{1}{n}w_n(\theta) = \frac{1}{n}\log W_n(\theta)$ to an e-power function $M(\theta)$ with a well-separated global minimum at $M(\theta^*) = 0$.

\begin{theorem}[Consistency of betting ME-estimators] \label{thm:consistency} Let $\Theta = [0,1]$ and $\lambda_t(\theta)$ be a predictable strategy such that $\frac{1}{n} w_n(\theta)$ converges uniformly in probability to a continuous limit function $M(\theta) = \Exp_{\theta^*}[\log(1 + \lambda^*(\theta)(X - \theta))]$. 
If the asymptotic bet fraction $\lambda^*(\theta)$ satisfies $\lambda^*(\theta^*) = 0$ and creates uniformly positive expected growth
\begin{equation}
    \inf_{ |\theta -\theta^*| \ge \varepsilon } \Exp_{\theta^*}[\log(1 + \lambda^*(\theta)(X - \theta))] > 0
\end{equation}
for all $\varepsilon > 0$, then any sequence of ME-estimators $\tilde{\theta}_n^{(E)} \in \argmin_{\theta \in \Theta} w_n(\theta)$ is consistent:
\[ \tilde{\theta}_n^{(E)} \xrightarrow{p} \theta^* .\]
\end{theorem}

\begin{proof} Note that $\sup_{\theta \in \Theta}|\frac{1}{n}w_n(\theta) - M(\theta)| \xrightarrow{p} 0$ and $ \inf_{ |\theta -\theta^*| \ge \varepsilon } M(\theta) > M(\theta^*) = 0$.
We invoke the classical M-estimation result by
\citet[Theorem 5.7]{van2000asymptotic} on ME-estimators $\tilde \theta_n^{(E)} \in \argmin \frac{1}{n}w_n(\theta)$ to see that $\tilde{\theta}_n^{(E)} \xrightarrow{p} \theta^* $.
\end{proof}

The theorem holds for ME-estimators from KT and aGRAPA.

\begin{corollary}\label[corollary]{cor:consistency-kt-agrapa}
     KT and aGRAPA satisfy \cref{thm:consistency}'s assumptions, and their ME-estimators are consequently consistent.
\end{corollary}

We leave the detailed verification to \cref{sec:pf-consistency-kt-agrapa}, which envolves three steps: (1) uniform convergence of the bet fractions $\lambda_t(\theta)$ to $\lambda^*(\theta)$, 
 (2) uniform convergence of the sequential log-wealth $\frac{1}{n} w_n(\theta)$ to $M(\theta)$, and (3) the strict positivity condition $M(\theta) > \delta > 0$ for all $\theta$ such that $|\theta - \theta^*| \ge \varepsilon$.

It is also worth noting that the consistency result here does not pose any smoothness condition on $w_n(\theta)$. This enables a variation of the bounded mean betting \cref{ex:betting}, namely \emph{quantile} betting \citep[Section 7]{waudby2020estimating}. Take the median as an example, instead of the wealth process
\begin{equation}
    \prod_{t=1}^n(1 + \lambda_t(\theta) \cdot (X_t - \theta)),
\end{equation}
we now consider
\begin{equation}
    \prod_{t=1}^n(1 + \lambda_t(\theta) \cdot (\id_{ \{ X_t\le \theta \} } - 0.5)).
\end{equation}
Despite their similar forms, the quantile betting wealth is not continuous with respect to $\theta$, creating a major difference for minimization. Let us write down the corresponding KT betting strategy and the consistency result, omitting its straightforward proof.

\begin{example}[KT betting for median] \label[example]{ex:kt-quantile} Let $(X_n) \iid P^*$ where $P^*$ is a distribution on $\mathbb R$ with CDF $F^*$.  Then, the KT quantile betting log-wealth for $c \ge 1$
    \begin{gather}
    w_n(\theta) = \sum_{t=1}^n \log(1 + \lambda_t(\theta) \cdot (\id_{ \{ X_t\le \theta \} } - 0.5)), \\ \lambda_t(\theta) = \frac{1/2 + \sum_{i=1}^{t-1} (\id_{ \{ X_i\le \theta \} } - 0.5)}{tc}, 
\end{gather}
satisfies $\sup_{\theta \in \mathbb R}| \lambda_t(\theta) -  \lambda^*(\theta)  | \xrightarrow{p} 0$ and $\sup_{\theta \in \mathbb R}|\frac{1}{n}w_n(\theta) - M(\theta)| \xrightarrow{p} 0$, where
\begin{gather}
   \lambda^*(\theta) = c^{-1} (F^*(\theta) - 0.5) ,\\
    M(\theta) = F^*(\theta) \log(1 + (2c)^{-1} (F^*(\theta) - 0.5) ) + (1-F^*(\theta)) \log (1 - (2c)^{-1} (F^*(\theta) - 0.5) ).
\end{gather}
Additionally, if $F^*$ is continuous at the median $F^*(\theta^*) = 0.5$, then $M(\theta^*) = 0$ is a well-separated minimum of $M(\theta)$. Consequently, any sequence of ME-estimators
\begin{equation}
    \tilde\theta_n^{(E)} \in \argmin_{X_{(k)}} w_n(X_{(k)}) \subseteq \argmin_{\theta \in \mathbb R} w_n(\theta)
\end{equation}
converges in probability to the true median $\theta^*$. Here, for $k = 0,1,\dots, n$, we denote by $X_{(k)}$ the $k$\textsuperscript{th} ascending order statistic, with the convention $X_{(0)} := X_{(1)} - 1$.
\end{example}

\subsection{Asymptotic normality and relative efficiency}

Our analysis on the asymptotic normality of ME-estimators $\tilde{\theta}_n^{(E)}$ follows a path similar to the classical analysis of M-estimators. We apply the first-order condition on the local minima of an adequately smooth $w_n(\theta) = \log W_n(\theta)$ function to see that $S_n(\tilde{\theta}_n^{(E)})  = 0$, where
\begin{equation}\label{eqn:score}
    S_n(\theta) := w_n'(\theta) = \sum_{t=1}^n\frac{\lambda_t'(\theta)(X_t - \theta) - \lambda_t(\theta)}{1 + \lambda_t(\theta)(X_t - \theta)}
\end{equation}
is the ``e-score function'', which is an analog of the score function in parametric models. Then, via Taylor expansion,
\begin{equation}
    0 = S_n(\tilde{\theta}_n^{(E)}) =  S_n(\theta^*) + (\tilde{\theta}_n^{(E)} - \theta^* ) S_n'(\theta^*) + \frac{1}{2} (\tilde{\theta}_n^{(E)} - \theta^* )^2 S_n''(\bar \theta_n)
\end{equation}
for some $\bar \theta_n$ on the line segment between $\theta^*$ and $\tilde{\theta}_n^{(E)}$. Therefore, analogous to the standard heuristics for M-estimator's CLT (e.g.\ Chapter 5.3 in \cite{van2000asymptotic}), we have the representation
\[ \sqrt{n}\big(\tilde{\theta}_n^{(E)} - \theta^*\big) = - \left( \frac{1}{n}S_n'(\theta^*) + \frac{1}{2n} S_n''(\bar{\theta}_n)(\tilde{\theta}_n^{(E)} - \theta^*) \right)^{-1} \left( \frac{1}{\sqrt{n}} S_n(\theta^*) \right). \]
Intuitively:
\begin{enumerate}
    \item \textbf{$\frac{1}{\sqrt{n}} S_n(\theta^*) $ satisfies a CLT, because $\lambda_t(\theta^*)$ and $\lambda_t'(\theta^*)$ both ``stabilize'', leaving the sum \eqref{eqn:score} asymptotically a sum of i.i.d.\footnote{We shall see the nuance in \cref{sec:pf-clt} that it actually cannot be approximated by the sum of i.i.d.\ random variables but is instead approximated by the row sums of a \emph{triangular} array of indepedent and \emph{almost} identically distributed random variables (distributions of the RVs in a row have the same shape, but ``stretched'' logarithmically)}\ random variables;} 
    \item \textbf{$\frac{1}{n}S_n'(\theta^*)$ satisfies a LLN, with similar reasoning as above;}
    \item \textbf{$\frac{1}{2n} S_n''(\bar{\theta}_n)(\tilde{\theta}_n^{(E)} - \theta^*) $ vanishes by controling $|S_n''/n|$ and the consistency of $\tilde{\theta}_n^{(E)}$.}
\end{enumerate}
Therefore, $\sqrt{n}\big(\tilde{\theta}_n^{(E)} - \theta^*\big) $ should also satisfy a CLT.

Here, the only things that matter should be \emph{the local behavior of $\lambda_t(\theta)$, $\lambda_t'(\theta)$, and $\lambda_t''(\theta)$ near $\theta \approx \theta^*$}.
Let us carry out the computation
for KT

\[ \lambda_t^{KT}(\theta) = \frac{1/2 + \sum_{i=1}^{t-1} (X_i - \theta)}{t c }, \]
where $c \ge 1/0.6$ is some constant; and aGRAPA

$$\lambda_t^{aG}(\theta) = \frac{\bar{X}_{t-1} - \theta}{\hat{\sigma}_{t-1}^2 + (\bar{X}_{t-1} - \theta)^2} \wedge 0.6 \vee -0.6 .$$

It turns out, from the asymptotic normality standpoint, one may see aGRAPA at $\theta^*$ as equivalent to KT with $c=\sigma^2$. ME-estimators from ``approximately KT'' betting strategies (which include KT and aGRAPA) satisfy the following central limit theorem.

\begin{theorem}[Asymptotic normality of betting ME-estimators]\label{thm:clt-me}
Let $X_1, X_2, \dots \in [0,1]$ be i.i.d.\ with mean $\theta^*$ and variance $\sigma^2 > 0$. Let
\[
  W_n(\theta) = \prod_{t=1}^n \bigl(1 + \lambda_t(\theta)(X_t - \theta)\bigr), \qquad
  w_n(\theta) = \sum_{t=1}^n \log\bigl(1 + \lambda_t(\theta)(X_t - \theta)\bigr),
\]
and let ME-estimators be defined as $\tilde\theta_n^{(E)} \in \argmin_{\theta\in\Theta} w_n(\theta)$, where $\Theta = [0,1]$ contains $\theta^*$ in its interior $(0,1)$. Assume $\tilde\theta_n^{(E)} \xrightarrow{p} \theta^*$ (as established by \cref{thm:consistency}).
Furthermore, suppose the predictable betting strategy $\lambda_t(\theta)$ is twice continuously differentiable with respect to $\theta$ almost surely, and satisfies the following local conditions at or around $\theta^*$:
\begin{enumerate}
  \item[\textnormal{(A1)}] \textbf{Stabilization of derivatives at $\theta^*$.}
  There exists a constant $\beta > 0$ such that
  \[
    \lambda_t'(\theta^*) \xrightarrow{L^2} -\beta \qquad \text{and} \qquad \frac{1}{n}\sum_{t=1}^n |\lambda_t''(\theta^*)| = O_p(1).
  \]

  \item[\textnormal{(A2)}] \textbf{Approximately KT fractions at $\theta^*$.}
  The betting fractions evaluated at the true parameter match the estimation error of the empirical mean in the first order:
  \[
      \lambda_t(\theta^*) = \beta(\bar X_{t-1} - \theta^*) + r_t,
  \]
  where $\bar X_{t-1} = \frac{1}{t-1}\sum_{i=1}^{t-1} X_i$ (with $\bar X_0 := \theta^*$), and the residual $r_t$ satisfies $n^{-1/2} \sum_{t=1}^n |r_t| \xrightarrow{p} 0$.

  \item[\textnormal{(A3)}] \textbf{Local smoothness.}
  There exists a neighborhood $\Theta_0$ of $\theta^*$ and a constant $L < \infty$ such that
  \[
    \sup_{\theta\in\Theta_0} \bigl( |\lambda_t(\theta)| + |\lambda_t'(\theta)| + |\lambda_t''(\theta)| \bigr) \le L \qquad \text{a.s.\ for all } t.
  \]
  Additionally, we assume the bets are bounded away from $-1$ in this neighborhood, i.e., $1 + \lambda_t(\theta)(X_t - \theta) \ge \epsilon > 0$.
\end{enumerate}
Then, ME-estimators are asymptotically normal:
\[
   \sqrt{n}\bigl(\tilde\theta_n^{(E)} - \theta^*\bigr)
   \xrightarrow{d} \mathcal N\bigl(0, v_\beta\bigr),
\]
where the asymptotic variance is given by
\begin{equation}\label{eqn:vbeta}
   v_\beta = \sigma^2 \frac{5 - 6\beta\sigma^2 + 2\beta^2\sigma^4}{(2 - \beta\sigma^2)^2}, \qquad \text{for } 0 < \beta < \frac{2}{\sigma^2}.
\end{equation}
\end{theorem}

\begin{corollary}
    Any sequence of ME-estimators from KT with parameter $c \ge 1/0.6$ satisfies \cref{thm:clt-me}'s assumptions with $\beta = 1/c$. The asymptotic variance $v_{1/c} \ge \sigma^2$ and $v_{1/c} \to 1.25 \sigma^2$ when $c \to \infty$. Any sequence of ME-estimators from aGRAPA satisfies \cref{thm:clt-me}'s assumptions with $\beta = 1/\sigma^2$. The asymptotic variance $v_{1/\sigma^2} = \sigma^2$. That is, aGRAPA's ME-estimator is efficient.
\end{corollary}

The proof of \cref{thm:clt-me}, whose high-level roadmap we already set out at the beginning of this subsection, is spelled out fully in \cref{sec:pf-clt}. 
The $\approx$1.25 asymptotic relative efficiency of KT's ME-estimators with large $c$, the asymptotic efficiency of aGRAPA's ME-estimators, as well as a simple normality check, are demonstrated by the following experiment. We compute a ME-estimator (a local minimum computed by Brent's method) and the vanilla sample mean on a sample of $n=100,000$ i.i.d.\ observations uniformly distributed on $[0,1]$,
with 2,500 repetitions for evaluating the sampling distribution.

\begin{center}
\begin{tabular}{lccc}
\toprule
Estimator & Theory Ratio ($v_\beta/\sigma^2$) & Simulation Ratio& JB Test ($p$-value) \\
\midrule
Sample Mean & 1.0000 & 1.0000 & 0.1429 \\
KT ($c=10$) & 1.2479 & 1.2655 & 0.1227 \\
aGRAPA & 1.0000 & 1.0012 & 0.1457 \\
\bottomrule
\end{tabular}
\end{center}
 The observed variance ratios show an excellent match with our theory. Furthermore, the non-rejecting Jarque-Bera test $p$-values corroborate the asymptotic normality of these ME-estimators.

Finally, we remark that we do not attempt at the asymptotic normality of the quantile betting scheme \cref{ex:kt-quantile}, as the required thrice differentiable log-wealth $w_n$ in bounded mean betting deteriorates drastically to a discontinuous function. The situation is very different from that of the classical M-estimator formulation of the sample quantile, where the minimized objective function is still Lipschitz, allowing methods like that of \citet[Theorem 5.23]{van2000asymptotic}.

\section{Discussions}

\subsection{Mixture betting and saddle point}

We have left out an important class of bounded betting strategies from our exposition above: mixture betting. Let $\pi$ be a probability measure on $[-1,1]$. The mixture betting wealth is
\begin{equation}\label{eqn:mix}
    W_n^\theta = \int \prod_{t=1}^n (1 + \lambda (X_t - \theta)) \pi(\d \lambda).
\end{equation}
Mixture betting dates as far back as the ``universal portfolio'' studied by \cite{cover1991universal}, and there have been multiple results on the pathwise regret of mixture betting \citep{cover2002universal,orabona2023tight}: some specific choices of $\pi$ lead to
\begin{equation}\label{eqn:regret}
    \log W_n^\theta \ge  \sup_{\lambda \in [-\frac{1}{1-\theta}, \frac{1}{\theta}] }\left\{ \sum_{t=1}^n \log(1 + \lambda (X_t - \theta)) \right\} - \operatorname{Reg}_n,
\end{equation}
which holds for \emph{any} sequence $X_1, X_2,\dots \in [0,1]$,
where the regret $\operatorname{Reg}_n$ is a non-random sequence (e.g.\ $\operatorname{Reg}_n = \log \sqrt{\pi (n+1)}$, see \cite{orabona2023tight}) independent of $\theta$.

These regret bounds \eqref{eqn:regret} conveniently resolve the issue of computability of the mixture integral \eqref{eqn:mix}: instead of approximating the integral \eqref{eqn:mix}, it is a standard practice in sequential statistics literature to use the right hand side of \eqref{eqn:regret}, which is a log-e-process since it lower-bounds the log-e-process on the left hand side, to construct sequential tests and confidence sets.

We therefore study ME-estimators derived from this regret log-e-process, which takes the following ``saddle point'' minimax form
\begin{align}\label{eqn:saddle}
    \tilde\theta^{(E)}_n &\in \argmin_{\theta \in [0,1] }\left\{ \sup_{\lambda \in [-\frac{1}{1-\theta}, \frac{1}{\theta}] } \sum_{t=1}^n \log(1 + \lambda (X_t - \theta)) - \operatorname{Reg}_n \right\} \\ &= \argmin_{\theta \in [0,1] }\sup_{\lambda \in [-\frac{1}{1-\theta}, \frac{1}{\theta}] }\left\{ \sum_{t=1}^n \log(1 + \lambda (X_t - \theta)) \right\} = \argmin_{\theta \in [0,1] } \klinf(\theta).
\end{align}
The last equality above, we note, comes from the fact that the supremum here is known as the $\klinf$ statistic (multiplied by a constant $n$) studied by e.g.\ \cite{honda2010asymptotically,thesis}. $\klinf$ is also related to the empirical likelihood and dual likelihood \citep{Owen1990,mykland1995dual} well-studied in the nonparametric literature, see
\citet[Section 5.1]{wang2026almost} and our upcoming discussion in \cref{sec:el}.

The saddle point estimator \eqref{eqn:saddle}, as it turns out, is nothing but the sample mean.

\begin{theorem}[Sample mean as mixture/saddle point ME-estimator] \label{thm:saddle}
    Let $X_1, \dots, X_n \in [0,1]$ be any sequence of observations. The saddle point ME-estimator
    \begin{equation}\label{eqn:saddle-in-theorem}
        \tilde\theta^{(E)}_n = \argmin_{\theta \in [0,1] }\sup_{\lambda \in [-\frac{1}{1-\theta}, \frac{1}{\theta}] }\left\{ \sum_{t=1}^n \log(1 + \lambda (X_t - \theta)) \right\}
    \end{equation}
    is {exactly the sample mean}:
    \begin{equation}
        \tilde\theta^{(E)}_n = \bar{X}_n = \frac{1}{n}\sum_{t=1}^n X_t.
    \end{equation}
    Consequently, it is consistent $\tilde\theta^{(E)}_n \to \theta^*$ and efficient $\sqrt{n}(\tilde\theta^{(E)}_n - \theta^*) \stackrel d \to \mathcal{N}(0,\sigma^2)$ as the sample mean is.
\end{theorem}

The proof of \cref{thm:saddle} is in \cref{sec:pf-saddle} and involves simply checking the derivatives and extrema of the bivariate function $(\theta, \lambda)\mapsto \sum \log(1+\lambda(X_t - \theta))$. $(\theta, \lambda) = (\bar X_n, 0)$ is a saddle point of this function with global optimality \eqref{eqn:saddle-in-theorem} satisfied. Separately, one may also prove \cref{thm:saddle} via the ``primal'' form (see \cite{honda2010asymptotically}) of the $\klinf$ statistic, and argue that $\klinf(\theta) \ge 0$ with equality if and only if $\theta = \bar X_n$.

It is still of (at least) theoretical interest what properties ME-estimators derived \emph{directly} from the mixture, 
$\argmin_{\theta} \int  \prod_{t=1}^n (1 + \lambda (X_t - \theta)) \pi(\d \lambda) $,
satisfy, despite the likely incomputability. We conjecture that for well-behaved (e.g.\ continuous) mixing measures $\pi$, this ME-estimator is strongly consistent and relatively efficient.

\subsection{Buyer's optimization: $\klinf$, empirical and dual likelihoods}\label{sec:el}

We now reveal the deeper connection between our methods and $\klinf$, the empirical and dual likelihoods (EL \& DL) from the classical nonparametric literature \citep{Owen1990,mykland1995dual}, and discuss how all these methods jointly complete an open-ended game-theoretic picture.

The minimax saddle point form \eqref{eqn:saddle} of the mixture ME-estimator discussed above in fact introduces the other player of the game we described in \cref{sec:why-e-est} (``the game involves two parties, a seller and a buyer...'') to the estimation framework. To wit, the realized e-statistic $M_n^\theta$ is the bet payoff the seller pays post-observation to the buyer. The optimization for the general ME-estimator \eqref{eqn:e-estimator} $\operatorname{minimize}_{\theta \in \Theta} M_n^\theta$ corresponds to maximizing the seller's profit. Now, let us introduce another parameter $\lambda \in \Lambda$ to the e-statistics family, and assume that the doubly indexed family
\begin{equation}
    \{ M_n(\theta,\lambda) : \theta \in \Theta, \lambda \in \Lambda \}
\end{equation}
satisfies that, for all $\theta \in \Theta$, for all $\lambda \in \Lambda$, $M_n(\theta, \lambda)$ is an e-value or forms an e-process under $\Pr_{\theta}$. In the bounded betting game \cref{ex:betting} the fixed-fraction bets
\begin{equation}
    W_n(\theta, \lambda) = \prod_{t=1}^n (1 + \lambda(X_t - \theta))
\end{equation}
for such a family. Importantly, the seller picks $\theta \in \Theta$ to decide the bets to offer, and the buyer picks $\lambda \in \Lambda$ to decide on one bet among many against $\theta$.

The minimax saddle point formulation of the mixture ME-estimator \eqref{eqn:saddle}, as it now transpires, corresponds to a Nash equilibrium of this two-player game where the buyer and the seller both see the observations and pick their parameter in sync. As a side result of independent interest, we show in \cref{sec:subpsi} that for all sub-$\psi$ e-processes (which include a wide variety of nonparametric problems, see \cite{howard2020time}), the saddle point is always the sample mean.

We are now in a good position to discuss the $\klinf$, EL, and DL. In the bounded betting case, these methods all concern \emph{the buyer's optimization problem after the seller has picked $\theta \in \Theta$}. That is, the buyer, knowing the (null hypothesis) $\theta$ picked by the seller, picks the oracle best-in-hindsight $\lambda$ to maximize the profit $M_n(\theta,\lambda)$. This is a game-theoretic problem orthogonal to the one that underpins ME-estimators, where the seller minimizes the e-statistic while knowing the buyer's move. To see this, the $\klinf$, EL, and DL statistics all take the form of
\begin{equation}
    \max_{\lambda} \sum_{t=1}^n \log(1+\lambda(X_t - \theta)) = \log \left\{\max_{\lambda} W_n(\theta, \lambda)\right\},
\end{equation}
differing only in the ranges of $\lambda$. See \citet[Section 5.1]{wang2026almost} for more on this distinction. Together with the seller's minimization problem and both players' equilibrium problem, it completes a picture of three possible optimization perspectives that the e-statistics game offers. We summarize them in \cref{tab:games}.

\begin{table}[!h]
    \centering
    \begin{tabular}{c|c|c}
        Optimizing player & Formula & Statistical method  \\
        \hline
        Buyer  & $\operatorname{maximize}_{\lambda \in \Lambda} M_n(\theta, \lambda) $ & EL, DL, and $\klinf$ \\
        Seller  &  $\operatorname{minimize}_{\theta \in \Theta} M_n(\theta, \lambda) $ & ME-estimators \\
        Both (equilibrium) & $\min_{\theta \in \Theta} \max_{\lambda \in \Lambda} M_n(\theta, \lambda)$  & Sample mean (betting \& sub-$\psi$)
    \end{tabular}
    \caption{Three optimization problems in the e-statistics game}
    \label{tab:games}
\end{table}

\subsection{Some other related work}

There are a few other pieces of work related to ME-estimators.

First, the idea of defining a point estimator from a family of tests and the corresponding confidence set was previously explored by \cite{birge2006model} as the ``T-estimator'' as a robustified MLE. While both the ME- and T-estimator take on a ``test first, point estimator second'' approach, they are constructed on very different metrics (e-statistics vs.\ penalized distance).

Second, several variants of the ``sizing the confidence sets based on the uniform growth of e-statistics'' argument exist in the literature. \citet[Theorems 5 and 6]{wasserman2020universal} size the confidence sets derived from the universal inference (i.e.\ split-sample likelihood ratio) e-values,
based on
``uniform upper bounds on the likelihood ratio'' (which are lower bounds on the e-statistics). A recent result on the same subject (albeit in the opposite direction, on the largeness instead of the smallness of the confidence sets) is due to \cite{takatsu2025precise}. Another recent study by \citet[Theorem 4.6]{shekhar2023near} on the confidence sets derived from the $\klinf$-bearing betting e-processes
 (i.e.\ right hand side of \eqref{eqn:regret}) quantifies the uniform largeness of the $\klinf$ statistic. Note that the ME-estimator in this setting is the sample mean as demonstrated in \cref{thm:saddle}, bypassing the confidence set analysis.

\section{Summary}

We introduced ME-estimators, the point estimator that minimizes a family of e-statistics in hindsight. ME-estimators align with two fundamental understandings of e-statistics: a measure of evidence against the null hypothesis, and a form of betting payoff. We developed two approaches to analyze the theoretical properties of ME-estimators. In the first approach, we utilize the inherent safety (type 1 error) guarantee of e-statistics and relate the strong consistency and almost sure convergence rate of ME-estimators to the sizing of the confidence sets derived from thresholding e-statistics. In the second approach, we generalize classical analysis of M-estimators to ME-estimators from the bounded betting setup, arriving at the weak consistency and asymptotic normality of the betting ME-estimator, backed by experiments.

Our results justify the use of ME-estimators (and therefore the use of the e-statistics) as a point estimation framework that vies with many existing ones in the statistics literature. We expect future work to demonstrate further use cases of ME-estimators, beyond the bounded mean estimation examples we explored in detail here in this paper. In particular, there might be estimation problems where powerful e-statistics are much easier to construct than reasonable point estimators. These will likely be the scenarios where ME-estimators prevail the most.

\subsubsection*{Acknowledgments}

HW thanks Sivaraman Balakrishnan, Arun Kumar Kuchibhotla, Ian Waudby-Smith, Shubhada Agrawal, and Peter Gr\"unwald for helpful discussions.

\bibliography{main}

\newpage

\appendix

\section{Odds and ends}

\subsection{Saddle point with sub-$\psi$ e-processes}\label[appendix]{sec:subpsi}

Let us recall the definition of sub-$\psi$ processes from \cite{howard2020time}:
\begin{equation}
    \exp( \lambda H_n - \psi(\lambda) V_n ) \quad \text{is an e-process for any }\lambda \in [0, \lambda_{\max})
\end{equation}
where $\psi$ is ``CGF-like'' \citep[Definition 2]{howard2020time}. In this note, we require that $\psi$ is defined on $(-\lambda_{\min}, \lambda_{\max})$ that includes 0, and is CGF-like in the sense that both $\psi(x)$ on $[0,\lambda_{\max})$ and $\psi(-x)$ on $[0,\lambda_{\min})$ are. In fact, $\psi$ will often be even.

Let us also recall that $\psi^\star$ is the convex conjugate (Legendre-Fenchel transform) of $\psi$ defined as
\begin{equation}
    \psi^\star(u) = \sup_{\lambda \in (-\lambda_{\min}, \lambda_{\max}) }\{ \lambda u - \psi(\lambda) \}.
\end{equation}
It is easy to see that $\psi^\star$ is uniquely minimized at 0: ${\psi^\star}'(u) = {\psi'}^{-1}(u)$ so $\sg({\psi^\star}'(u)) = \sg(u)$.

In many of the examples enumerated by \citet[Table 3]{howard2020time}, we are testing the \emph{location} (not necessarily the mean) of i.i.d.\ $(X_n)$ is $\theta$ with $H_n(\theta) = \sum_{k=1}^n (X_k - \theta)$ and $V_n(\theta)$ depends on $(X_n)$ and $\theta$ (usually through $(X_n - \theta)$). For example, $V_n$ is usually a combination of
\begin{equation}
    n, \quad  \sum_{t=1}^n (X_t - \theta)^2, \quad \sum_{t=1}^n \Exp_{\theta}((X_t - \theta)^2).
\end{equation}
The third term above is unobservable and needs to be replaced by a reasonable upper bound and/or supplied by problem assumptions.

\begin{proposition}
    For the sub-$\psi$ location estimation problem where $H_n(\theta)  = \sum_{t=1}^n (X_t - \theta)$ and $V_n(\theta) > 0$ for a sub-$\psi$ pair for all $\theta$, the saddle point estimator always equals the sample mean:
    \begin{equation}
        \argmin_{\theta \in \mathbb R } \sup_{\lambda \in  (\lambda_{\min}, \lambda_{\max}) } \{\lambda H_n(\theta) - \psi(\lambda) V_n(\theta)\} = S_n/n
    \end{equation}
    where $S_n = X_1+\dots + X_n$.
\end{proposition}
\begin{proof}
Using (1)  $V_n(\theta) > 0$, (2) definition of $\psi^\star$, and (3) $\psi^\star$ is uniquely minimized at 0, we have
\begin{align}
   &  \argmin_{\theta \in \mathbb R } \sup_{\lambda } \left\{ \lambda(S_n - n\theta) - \psi(\lambda) V_n(\theta)  \right\},
   \\
   = & \argmin_{\theta \in \mathbb R } \sup_{\lambda } \left\{-\psi(\lambda)+ \frac{S_n - n\theta}{V_n(\theta)} \lambda   \right\},
   \\
    = & \argmin_{\theta \in \mathbb R } \psi^\star( (S_n -n\theta)/V_n(\theta) )
   \\
    = & S_n/n. \qedhere
\end{align}   
\end{proof}

In particular, we see that the expression of $V_n(\theta)$ does not matter.
Even when the sub-$\psi$ pair $(H_n(\theta), V_n(\theta))$ comes from mean estimation for heavy-tailed observations plus self-normalization, or median estimation for symmetric observations plus self-normalization, the estimator recommended by this game-theoretic heuristic is still the empirical mean $S_n/n$, which is suboptimal in these cases.

\section{Omitted proofs}

\subsection{Proof of \cref{thm:BC-epr}}\label[appendix]{sec:pf-bc-epr}

\begin{proof} 
Define the coverage events
\begin{equation}
    A_\alpha = \left\{  \sup_{n \ge 1}  M_n^{\theta^*}  \le 1/\alpha \right\}, \quad C_{\alpha, n} = \{ \theta^* \in  \CS_n^{\alpha} \} = \{ M^{\theta^*}_n \le 1/\alpha  \}.
\end{equation}
Then,
\begin{equation}
    \Pr_{\theta^*}( A_\alpha) \ge 1-\alpha, \quad  A_\alpha = \bigcap_{n=1}^\infty C_{\alpha, n},
\end{equation}
and for any $\alpha > \alpha'$,
\begin{equation}
    A_\alpha  \subseteq A_{\alpha'}, \quad C_{\alpha, n}  \subseteq C_{\alpha', n}.
\end{equation}
Now, consider the coverage events for the doubling scheme
\begin{equation}
    C_{k}^* = \left\{  \theta^* \in \bigcap_{n \ge n_k} \CS_n^* \right\} = \bigcap_{s=k}^\infty \bigcap_{n\in [n_s, n_{s+1})} C_{\alpha_s, n}.
\end{equation}
Note that, for any $s \ge k$ and $n \ge n_s$, $C_{\alpha_s, n} \supseteq C_{\alpha_k, n} \supseteq A_{\alpha_k}$. We see that $C_{k}^* \supseteq A_{\alpha_k}$. Therefore, we have constructed a nested set sequence
\begin{equation}
    C_1^* \subseteq C_2^* \subseteq \dots
\end{equation}
with $\Pr_{\theta^*}(C_k^*) \ge 1-\alpha_k$. Consider
\begin{equation}
    C^* = \bigcup_{k=1}^\infty C_k^* = \left\{ \text{there exists } k\ge 1 \text{ such that }  \theta^* \in \bigcap_{n \ge n_k} \CS_n^* \right\} 
\end{equation}
Then $\Pr_{\theta^*}(C^*) = 1$.
 \end{proof}

\subsection{Proof of \cref{ex:det-hedge}} \label[appendix]{sec:pf-det-hedge}

Note that $\sum_{k=1}^n \lambda_k^2 = \frac{1}{3}H_n \le \frac{\log n + 1}{3}$, and $\sum_{k=1}^n \lambda_k \ge \frac{2}{\sqrt{3}}\sqrt{n}$ since
\begin{equation}
    \sqrt{k} - \sqrt{k-1} = \frac{1}{\sqrt{k} + \sqrt{k-1}} \le \frac{1}{2\sqrt{k}}.
\end{equation}

We study the uniform growth of the family $\{  W_n^{H(\blamb)} (\theta) : \theta \neq \theta^* \}$ under $P$. To achieve this, for any $\theta < \theta^*$:
\begin{equation}
   W_n (\theta) \ge  \frac{1}{2} \prod_{k=1}^n (1 + \lambda_k(X_k - \theta) ).
\end{equation}
Therefore, noting that $\log(1+x) \ge x-x^2$ for all $|x| \le 1/\sqrt{3}$,
\begin{align}
   & \log W_n (\theta) \ge - \log 2 +  \sum_{k=1}^n \log( 1 + \lambda_k (X_k - \theta) )
    \\
   \ge & - \log 2 +  \sum_{k=1}^n \lambda_k (X_k - \theta) -  \sum_{k=1}^n \lambda_k^2 (X_k - \theta)^2
   \\
   \ge  & - \log 2 +  \sum_{k=1}^n \lambda_k (X_k - \theta) -  \sum_{k=1}^n \lambda_k^2
   \\
   \ge  & - \log 2 - \frac{\log n+1}{3} +  \sum_{k=1}^n \lambda_k (X_k - \theta) 
    \\
    = & - \log 2 - \frac{\log n+1}{3} + { \sum_{k=1}^n \lambda_k (\theta^* - \theta) } +  {\sum_{k=1}^n \lambda_k (X_k - \theta^*) }
    \\
    \ge & - \log 2 - \frac{\log n+1}{3} + {  2\sqrt{n/3}(\theta^* - \theta) } +  { \sum_{k=1}^n \lambda_k (X_k - \theta^*) }
\end{align}
Above, the third term is the first-order nonrandom ``signal'' term that drives the wealth growth;
the fourth term is the martingale random fluctuation term which we control by a standard Hoeffding:
\begin{equation}
    P^*\left( \sum_{k=1}^n \lambda_k (X_k - \theta^*) \le - t \right) \le \exp\left( -\frac{2t^2}{\sum_{k=1}^n \lambda_k^2 } \right) \le \exp\left( -\frac{6t^2}{\log n + 1} \right).
\end{equation}
In sum,
\begin{equation}
    P^* \left( \forall \theta \in [0,\theta^*),\  \log W_n(\theta) \ge  - \log 2 - \frac{\log n+1}{3} +  2\sqrt{n/3}(\theta^* - \theta) - t  \right) \ge 1 - \exp\left( -\frac{6t^2}{\log n + 1} \right).
\end{equation}
A simple union bound gives
\begin{equation}
    P^* \left( \forall \theta \neq \theta^*,\  \log W_n(\theta) \ge  - \log 2 - \frac{\log n+1}{3} +  2\sqrt{n/3}|\theta^* - \theta| - t  \right) \ge 1 - 2\exp\left( -\frac{6t^2}{\log n + 1} \right).
\end{equation}
Taking $2\exp\left( -\frac{6t^2}{\log n + 1} \right)  =\delta$, we arrive at
\begin{equation}
  P^* \left( \forall \theta \neq \theta^*,\  \log W_n(\theta) \ge 2\sqrt{n/3}|\theta^* - \theta| - \log 2 - \frac{\log n+1}{3}   - \sqrt{\frac{\log(2en/\delta)}{6}} \right) \ge 1 - \delta.
\end{equation}
Consequently,
\begin{equation}
    P^*\left( |\CS_n^{\alpha}| \le  \frac{\log(1/\alpha) + \log 2 + \frac{\log n + 1}{3} + \sqrt{\frac{\log(2en/\delta)}{6}}  }{\sqrt{n/3}}  \right) \ge 1-\delta,
\end{equation}
and
\begin{equation}
     \tilde \theta_n^{(E)} = \theta^* + \mathcal{O}_{a.s.}(n^{-1/2}\log n )
\end{equation}

\subsection{Proof of \cref{ex:kt-ugi}} \label[appendix]{sec:pf-kt}

\begin{proof}[Proof of \cref{ex:kt-ugi}]
Using $\log(1+x) \ge x - x^2$ for $x \ge -1/2$, we have
\begin{equation}\label{eq:kt-log-lower}
\log W_n^\theta \ge \sum_{t=1}^n \lambda_t^{KT}(\theta) (X_t - \theta) - \sum_{t=1}^n \left(\lambda_t^{KT}(\theta)\right)^2 (X_t - \theta)^2.
\end{equation}
We will now decouple $\theta$ from the stochastic terms. Let $Y_t = X_t - \theta^*$ and $\Delta = \theta^* - \theta$. Let $H_t = (1/2 - \theta^*) + \sum_{i=1}^{t-1} Y_i$, so that $\lambda_t^{KT}(\theta) = \frac{H_t}{2t} + \frac{\Delta}{2}$. 

Notice that $X_t - \theta = Y_t + \Delta$. For the linear term in \eqref{eq:kt-log-lower}, we have:
\begin{equation*}
\lambda_t^{KT}(\theta) (X_t - \theta) = \left( \frac{H_t}{2t} + \frac{\Delta}{2} \right) (Y_t + \Delta) = \frac{H_t Y_t}{2t} + \frac{H_t \Delta}{2t} + \frac{\Delta Y_t}{2} + \frac{\Delta^2}{2}.
\end{equation*}
For the quadratic term, since $(X_t - \theta)^2 \le 1$, we can bound it as:
\begin{equation*}
\left(\lambda_t^{KT}(\theta)\right)^2 (X_t - \theta)^2 \le \left( \frac{H_t}{2t} + \frac{\Delta}{2} \right)^2 = \frac{H_t^2}{4t^2} + \frac{H_t \Delta}{2t} + \frac{\Delta^2}{4}.
\end{equation*}
Subtracting the two bounds, we have,
\begin{align*}
\log W_n^\theta &\ge \sum_{t=1}^n \left( \frac{\Delta^2}{4} + \Delta \frac{Y_t}{2} + \frac{H_t Y_t}{2t} - \frac{H_t^2}{4t^2} \right) \\
&= \frac{n}{4} \Delta^2 + \Delta \left( \frac{1}{2} \sum_{t=1}^n Y_t \right) + \sum_{t=1}^n \frac{H_t Y_t}{2t} - \sum_{t=1}^n \frac{H_t^2}{4t^2}.
\end{align*}

We now apply concentration inequalities to the three stochastic components with probability $1-\delta/3$ each. 
First, $Y_t$ is a martingale difference sequence with $Y_t \in [-\theta^*, 1-\theta^*]$ of range $1$. By the Azuma-Hoeffding inequality, with probability $\ge 1 - \delta/3$,
\begin{equation*}
|A| := \left| \frac{1}{2} \sum_{t=1}^n Y_t \right| \le \frac{1}{2} \sqrt{\frac{n}{2} \log \frac{6}{\delta}} =: K_1.
\end{equation*}

Second, $H_t$ is the sum of independent bounded variables. By Hoeffding's inequality and a union bound, with probability $\ge 1 - \delta/3$, we have $|H_t| \le 1 + \sqrt{\frac{t}{2} \log \frac{6n}{\delta}}$ simultaneously for all $t \in [1, n]$. On this event, we can deterministically bound $C := \sum_{t=1}^n \frac{H_t^2}{4t^2}$:
\begin{equation*}
C \le \sum_{t=1}^n \frac{1}{4t^2} \left( 1 + \sqrt{\frac{t}{2} \log \frac{6n}{\delta}} \right)^2 \le \sum_{t=1}^n \frac{2}{4t^2} + \sum_{t=1}^n \frac{t}{8t^2} \log \frac{6n}{\delta} \le \frac{\pi^2}{12} + \frac{1}{8} \left(\log \frac{6n}{\delta}\right) (\log n + 1) =: C_n.
\end{equation*}

Third, $B := \sum_{t=1}^n \frac{H_t Y_t}{2t}$ is a martingale. Let us bound the stopped sum $B^\tau := \sum_{t=1}^{n\wedge \tau} \frac{H_t Y_t}{2t}$ where $\tau$ is last time before $H_t$ escapes its high-probability bounds proved above
\begin{equation}
    \tau = \max\left\{ t : |H_t| \le 1 + \sqrt{\frac{t}{2}\log \frac{6n}{\delta}} \right\},
\end{equation}
which is a stopping time due to the predictability of $(H_t)$.
Applying the Azuma-Hoeffding inequality to the stopped martingale prior to the escaping of $(H_t)$, its increments are bounded by $c_t = \frac{1 + \sqrt{\frac{t}{2} \log \frac{6n}{\delta}}}{2t}$, and $\sum_{t=1}^n c_t^2 \le C_n$. Thus, with probability $\ge 1 - \delta/3$, we have $|B^\tau| \le \sqrt{ \frac{C_n}{2} \log \frac{6}{\delta} }$.

Intersecting these three events, with probability $\ge 1-\delta$ uniformly over $\theta$, we obtain:
\begin{equation*}
\log W_n^\theta \ge \frac{n}{4} |\theta - \theta^*|^2 - K_1 |\theta - \theta^*| - \left( \sqrt{ \frac{C_n}{2} \log \frac{6}{\delta} } + C_n \right).
\end{equation*}
This establishes \eqref{eqn:kt-ugi}, noting $K_1 = Kn^{1/2} = \sqrt{\frac{n}{8} \log \frac{6}{\delta}}$. By letting $D_n(x, \delta) = \frac{n}{4} x^2 - K_1 x - K_2$, the inverse $\varphi_n(h, \delta)$ is given by the positive root $\frac{K_1 + \sqrt{K_1^2 + n(K_2 + h)}}{n/2} \le \frac{2K_1}{n/2} + \frac{\sqrt{n(K_2+h)}}{n/2}$, where $K_2 = \sqrt{ \frac{C_n}{2} \log \frac{6}{\delta} } + C_n$.
Using \cref{thm:master}, replacing $h$ by $\log(1/\alpha)$ yields the high-probability bound on the diameter of $\CS_n^\alpha$. Substituting $\alpha_n = n^{-1-\varepsilon}$ and $\delta_n = n^{-1-\varepsilon}$ yields $K_1 = \mathcal{O}(\sqrt{n \log n})$ and $K_2 = \mathcal{O}(\log^2 n)$, which gives $\diam(\CS_n^{\alpha_n}) = \mathcal{O}_{a.s.}(n^{-1/2} \log n)$, concluding the proof.
\end{proof}

\subsection{Proof of \cref{ex:pacbayes-sn}} \label[appendix]{sec:pf-pacbayes-sn}

\begin{proof} Recall that
    \begin{equation}
  \log  M_n^\theta = \left\| \sum_{t=1}^n \lambda_t (X_t - \theta) \right\| - \sum_{t=1}^n \frac{\lambda_t^2}{6}\left(\|   X_t - \theta \|^2(1+\beta^{-1}) + 2\|\Sigma\|_{\mathrm{op}} + \beta^{-1} \tr\Sigma  \right) - \frac{\beta}{2}.
\end{equation}
Under the ground truth $\theta^*$, we have
\begin{align}
    \left\| \sum_{t=1}^n \lambda_t (X_t - \theta) \right\| \ge   \left\| \sum_{t=1}^n \lambda_t (\theta^* - \theta) \right\| -  \left\| \sum_{t=1}^n \lambda_t (X_t - \theta^*) \right\|
\end{align}
The first part is deterministic, and the second part, decoupled from $\theta$, can be high-probability bounded by the log-e-process itself: with probability at least $1-\delta$,
\begin{equation}
    \left\| \sum_{t=1}^n \lambda_t (X_t - \theta^*) \right\| \le \sum_{t=1}^n \frac{\lambda_t^2}{6}\left(\|   X_t - \theta^* \|^2(1+\beta^{-1}) + 2\|\Sigma\|_{\mathrm{op}} + \beta^{-1} \tr\Sigma  \right) + \frac{\beta}{2} + \log(1/\delta).
\end{equation}
We further note that
\begin{equation}
    \| X_t - \theta \|^2 \le  2\| X_t - \theta^* \|^2 + 2\| \theta - \theta^* \|^2.
\end{equation}
Combining the bounds above, we have, with probability $1-\delta$, for all $\theta\in\mathbb R^d$:
\begin{align}
    & \log M_n^\theta \\
    \ge & \| \theta^* - \theta \|\sum_{t=1}^n \lambda_t - \sum_{k=1}^n \frac{\lambda_t^2}{6}\left(( 2\| X_t - \theta^* \|^2 + 2\| \theta - \theta^* \|^2)(1+\beta^{-1}) + 2\|\Sigma\|_{\mathrm{op}} + \beta^{-1} \tr\Sigma  \right) - \frac{\beta}{2} \\
    - & \sum_{t=1}^n \frac{\lambda_t^2}{6}\left(\|   X_t - \theta^* \|^2(1+\beta^{-1}) + 2\|\Sigma\|_{\mathrm{op}} + \beta^{-1} \tr\Sigma  \right) - \frac{\beta}{2} - \log(1/\delta)
    \\
    = & \| \theta^* - \theta \|\sum_{t=1}^n \lambda_t  - \| \theta^* - \theta \|^2 \frac{1+\beta^{-1}}{3} \sum_{t=1}^n \lambda_t^2 \\
    - & \sum_{t=1}^n \frac{\lambda_t^2}{6}\left( 3\|   X_t - \theta^* \|^2(1+\beta^{-1}) + 4\|\Sigma\|_{\mathrm{op}} + 2\beta^{-1} \tr\Sigma  \right) - \beta - 2\log(1/\delta).
\end{align}
With the additional assumption that $\|X_t\|\le L$:
\begin{align}
     & \log M_n^\theta \\
     \ge & \| \theta^* - \theta \|\sum_{t=1}^n \lambda_t  - 2L\cdot \| \theta^* - \theta \| \frac{1+\beta^{-1}}{3} \sum_{t=1}^n \lambda_t^2 \\
    - & \sum_{t=1}^n \frac{\lambda_t^2}{6}\left( 3\|   X_t - \theta^* \|^2(1+\beta^{-1}) + 4\|\Sigma\|_{\mathrm{op}} + 2\beta^{-1} \tr\Sigma  \right) - \beta - 2\log(1/\delta).
\end{align}
Therefore, by \cref{thm:master},
\begin{equation} \scriptsize
    P^* \left( \diam(\CS_n^{\alpha}) \le  \frac{ 2\sum_{t=1}^n \frac{\lambda_t^2}{6}\left( 3\|   X_t - \theta^* \|^2(1+\beta^{-1}) + 4\|\Sigma\|_{\mathrm{op}} + 2\beta^{-1} \tr\Sigma  \right) + 2\beta + 4\log(1/\delta) + 2\log(1/\alpha)   }{
        \sum_{t=1}^n \lambda_t - \frac{2L(1+\beta^{-1})}{3} \sum_{t=1}^n \lambda_t^2
    }\right)   \ge 1-\delta,
\end{equation}
Setting $\alpha = \delta = 1/n^2$, we see that almost surely, for all but finitely many $n$,
\begin{equation} 
   \scriptsize \| \theta^* - \tilde \theta_n^{(E)} \| \le\frac{ 2\sum_{t=1}^n \frac{\lambda_t^2}{6}\left( 3\|   X_t - \theta^* \|^2(1+\beta^{-1}) + 4\|\Sigma\|_{\mathrm{op}} + 2\beta^{-1} \tr\Sigma  \right) + 2\beta + 12\log n  }{
        \sum_{t=1}^n \lambda_t - \frac{2L(1+\beta^{-1})}{3} \sum_{t=1}^n \lambda_t^2
    }
\end{equation}
Taking $\lambda_t = 1/\sqrt{t}$, the right-hand side is $\mathcal{O}_{a.s.}(n^{-1/2}\log n)$, because the only randomness here satisfies
\begin{equation}
    \sum_{t=1}^n t^{-1} \| X_t - \theta^* \|^2 =  \sum_{t=1}^n t^{-1} (\| X_t - \theta^* \|^2 - \Exp^* \| X_t - \theta^* \|^2) + \Exp^* \| X_t - \theta^* \|^2 \sum_{t=1}^n t^{-1}  = \mathcal{O}_{a.s.}(\log n).
\end{equation}
Consequently, $\| \theta^* - \tilde \theta_n^{(E)} \| = \mathcal{O}_{a.s.}(n^{-1/2}\log n)$.
\end{proof}

\subsection{Proof of \cref{cor:consistency-kt-agrapa}} \label[appendix]{sec:pf-consistency-kt-agrapa}
\begin{proof}
    
To this end we must verify three conditions: (1) uniform convergence of the bet fractions $\lambda_t(\theta)$ to $\lambda^*(\theta)$, 
 (2) uniform convergence of the sequential log-wealth $\frac{1}{n} w_n(\theta)$ to $M(\theta)$, and (3) the strict positivity condition $M(\theta) > \delta > 0$ for all $\theta$ such that $|\theta - \theta^*| \ge \varepsilon$.

\paragraph{Step 1: Uniform convergence of $\lambda_t(\theta)$} First, let us first establish that $\lambda_t(\theta) \to \lambda^*(\theta)$ uniformly almost surely over $\theta \in \Theta$ for both KT and aGRAPA. The bets $\lambda_t^{KT}(\theta) = \frac{1/2 + \sum_{i=1}^{t-1} (X_i - \theta)}{t c}$ amd $\lambda_t^{aG}(\theta) = \frac{\bar{X}_{t-1} - \theta}{\hat{\sigma}_{t-1}^2 + (\bar{X}_{t-1} - \theta)^2}$ are both Lipschitz continuous functions on the sample mean, sample variance (which are consistent), and $\theta$ in a neighborhood containing $\{ \theta^* \} \times \{ \sigma^2 \} \times [0,1]$. Define $\lambda^{KT*}(\theta)= (\theta^* - \theta)/c$ and $\lambda^{aG*}= \frac{\theta^* - \theta}{\sigma^2 + (\theta^* - \theta)^2}$. We have,
\begin{gather}
    \sup_{\theta \in \Theta}| \lambda_t(\theta)  - \lambda^*(\theta)  | \lesssim |\bar{X}_{t-1} - \theta^* | + |\hat{\sigma}_{t-1}^2  - \sigma^2|\xrightarrow{a.s.} 0
\end{gather}
for both strategies.

\paragraph{Step 2: Uniform convergence of $\frac{1}{n}w_n(\theta)$}
With uniform convergence of $\lambda_t(\theta)$ in place, the uniform convergence of $\frac{1}{n}w_n(\theta)$ to $M(\theta)$ is straightforward. To see that,
\[ \frac{1}{n} w_n(\theta) = \frac{1}{n} \sum_{t=1}^n \big[ \log(1 + \lambda_t(\theta)(X_t - \theta)) - \log(1 + \lambda^*(\theta)(X_t - \theta)) \big] + \frac{1}{n} \sum_{t=1}^n \log(1 + \lambda^*(\theta)(X_t - \theta)) \]
 By Lipschitzness of $\log(1+x)$ on the bounded support $|x| \le 0.6$, the first difference term is bounded by $\sup_{\theta} |\lambda_t(\theta) - \lambda^*(\theta)| \cdot |X_t - \theta|$, which converges almost surely to $0$. Thus, the first summation vanishes uniformly. The second summation is a standard empirical average of i.i.d.\ variables evaluated through a continuous function of $\theta$. We now apply the uniform law of large numbers (e.g.\ Theorem 3.1(vi) in \cite{dasgupta2008asymptotic}) to the bounded, continuous function family
\begin{equation}
  g(x,\theta)= \log (1 + \lambda^*(\theta)(x - \theta)) 
\end{equation}
and obtain that
\begin{equation}
    \sup_{\theta \in \Theta} \left| \frac{1}{n} \sum_{t=1}^n \log(1 + \lambda^*(\theta)(X_t - \theta)) - \Exp_{\theta^*}\big[\log(1 + \lambda^*(\theta)(X - \theta))\big] \right| \xrightarrow{p} 0.
\end{equation}
This formally establishes the continuous limit function $M(\theta)$.

\paragraph{Step 3: Positivity} 
We apply the bound: $\log(1+x) > x - 0.9x^2$ for $|x| \le 0.6$ and $x \neq 0$:
\begin{equation}
    M(\theta)  > \Exp_{\theta^*}\left[ \lambda^*(\theta) (X - \theta)  -0.9{\lambda^{*2}(\theta) (X - \theta)^2} \right] = \lambda^*(\theta) (\theta^* - \theta)  - 0.9 \lambda^{*2}(\theta)\cdot (\sigma^2 + (\theta^* - \theta)^2).
\end{equation}

For aGRAPA, the limit fraction evaluates to $\lambda^*(\theta) = \frac{\theta^* - \theta}{\sigma^2 + (\theta^* - \theta)^2} \wedge 0.6 \vee -0.6 $. We see that 
\[ M(\theta) >  \lambda^*(\theta) (\theta^* - \theta)  -  {\lambda^{*2}(\theta)\cdot (\sigma^2 + (\theta^* - \theta)^2)} =  0.1{\lambda^{*2}(\theta)\cdot (\sigma^2 + (\theta^* - \theta)^2)}  \]
if $|\theta^* - \theta| \le 0.6(\sigma^2 + (\theta^* - \theta)^2)$;
\[
M(\theta) > 0.6|\theta^* - \theta| - 0.6^2( \sigma^2 + (\theta^* - \theta)^2 ) > 0.6(0.6( \sigma^2 + (\theta^* - \theta)^2 ))  - 0.6^2( \sigma^2 + (\theta^* - \theta)^2 ) = 0
\]
if $|\theta^* - \theta| > 0.6(\sigma^2 + (\theta^* - \theta)^2)$.

For KT, the asymptotic bet fraction is $\lambda^*(\theta) = (\theta^* - \theta)/c$. We see that
\[ M(\theta) >  \lambda^*(\theta) (\theta^* - \theta)  -  {\lambda^{*2}(\theta)\cdot (\sigma^2 + (\theta^* - \theta)^2)} = \frac{(\theta^* - \theta)^2}{c}\left( 1 - \frac{\sigma^2 + (\theta^* - \theta)^2}{c} \right)  \]
where $ 1 - \frac{\sigma^2 + (\theta^* - \theta)^2}{c} > 0$ as long as $c \ge 1/0.6$.

To conclude, for both KT and aGRAPA, we lower bound $M(\theta)$ by some function of $|\theta - \theta^*|$ with 0 as a well-separated minimum.
Thus, both KT and aGRAPA satisfy all requirements of \cref{thm:consistency}, guaranteeing consistency of their ME-estimators.
\end{proof}

\subsection{Proof of \cref{thm:clt-me}}\label[appendix]{sec:pf-clt}

When proving \cref{thm:clt-me}, we repeatedly implicitly invoke the following lemma: if a martingale's predictable quadratic variation is sublinear, then the martingale is $o_p(\sqrt{n})$.

\begin{lemma} \label{lem:martingale_op}
 Let $M_n = \sum_{t=1}^n X_t$ be a square-integrable martingale with predictable quadratic variation $\langle M \rangle_n = \sum_{t=1}^n \Exp[X_t^2 \mid \mathcal{F}_{t-1}]$.
If $\langle M \rangle_n = o_p(n)$, then $M_n = o_p(\sqrt{n})$.
\end{lemma}

\begin{proof}
We wish to show that for any fixed $\epsilon > 0$, 
\[ \lim_{n \to \infty} \Pr\Big( |M_n| \ge \epsilon \sqrt{n} \Big) = 0. \]
Fix a $\delta > 0$ and an $n \ge 1$. Since $\langle M \rangle$ is predictable,
\[ \tau_n = \inf\big\{ k \ge 0 : \langle M \rangle_{k+1} > \delta n \big\} \]
 is a stopping time.
Consider the stopped martingale $\tilde{M}_t = M_{t \wedge \tau_n}$. Its increments are $\tilde{X}_t = X_t \mathbbm{1}_{\{t \le \tau_n\}}$ and its predictable quadratic variation is
\[ \langle \tilde{M} \rangle_n = \sum_{t=1}^n \Exp\left[ \tilde{X}_t^2 \mid \mathcal{F}_{t-1} \right] = \sum_{t=1}^{n \wedge \tau_n} \Exp[X_t^2 \mid \mathcal{F}_{t-1}] = \langle M \rangle_{n \wedge \tau_n}. \]
By the definition of $\tau_n$, $\langle \tilde{M} \rangle_n = \langle M \rangle_{n\wedge \tau_n} \le \delta n$.
Consequently
$ \Exp\big[ \tilde{M}_n^2 \big] = \Exp\big[ \langle \tilde{M} \rangle_n \big] \le \delta n $, 
which via Chebyshev's inequality implies,
\[ \Pr\Big( |\tilde{M}_n| \ge \epsilon \sqrt{n} \Big) \le \frac{\Exp[\tilde{M}_n^2]}{\epsilon^2 n} \le \frac{\delta n}{\epsilon^2 n} = \frac{\delta}{\epsilon^2}. \]

Now, we relate the stopped martingale back to the original martingale $M$. Notice that $M_n=\tilde{M}_n$  on the event $\{\tau_n \ge n\}$. Therefore,
\begin{align*}
\Pr\Big( |M_n| \ge \epsilon \sqrt{n} \Big) 
&\le \Pr\Big( |\tilde{M}_n| \ge \epsilon \sqrt{n} \text{ and } \tau_n \ge n \Big) + \Pr(\tau_n < n) \\
&\le \Pr\Big( |\tilde{M}_n| \ge \epsilon \sqrt{n} \Big) + \Pr(\tau_n < n)
 \le \frac{\delta}{\epsilon^2} + \Pr\Big( \langle M \rangle_n > \delta n \Big)
\end{align*}
where the last step follows from $\{\tau_n < n\} \subseteq \{ \langle M \rangle_n > \delta n \}$. 
Taking limit and using the assumption $\langle M_n \rangle = o_p(n)$,
\[ \limsup_{n \to \infty} \Pr\Big( |M_n| \ge \epsilon \sqrt{n} \Big) \le \frac{\delta}{\epsilon^2}. \]
The arbitrariness of $\delta > 0$ (recall that we \emph{first} fixed $\epsilon$ \emph{then} fixed $\delta$) finally implies that
\[ \limsup_{n \to \infty} \Pr\Big( |M_n| \ge \epsilon \sqrt{n} \Big) = 0, \]
concluding the proof.
\end{proof}

We are now ready to prove \cref{thm:clt-me}.

\begin{proof}[Proof of \cref{thm:clt-me}]
Because $\tilde\theta_n^{(E)} \xrightarrow{p} \theta^*$ and $\theta^*$ is strictly in the interior of the compact set $\Theta$, the event $E_n = \{\tilde\theta_n^{(E)} \in \text{int}(\Theta)\}$ occurs with probability approaching 1 under $P^*$ as $n \to \infty$. Conditioned on $E_n$, the estimator $\tilde\theta_n^{(E)}$ is a local minimum, thus satisfying the first-order condition $S_n(\tilde\theta_n^{(E)}) = 0$. Using Taylor expansion on $E_n$, we have:
\begin{equation}\label{eqn:taylor-score}
   \sqrt{n}(\tilde\theta_n^{(E)} - \theta^*)
   = -\left(\frac{1}{n}S_n'(\theta^*) + \frac{1}{2n}S_n''(\bar\theta_n)(\tilde\theta_n^{(E)} - \theta^*)\right)^{-1}
      \left(\frac{1}{\sqrt{n}}S_n(\theta^*)\right),
\end{equation}
for some $\bar\theta_n$ between $\theta^*$ and $\tilde\theta_n^{(E)}$. Since $P^*(E_n) \to 1$, the asymptotic distribution of $\sqrt{n}(\tilde\theta_n^{(E)} - \theta^*)$ is entirely determined by the right hand side of \eqref{eqn:taylor-score}, since for all $x\in\mathbb R$:
\begin{equation}
   | P^*( \text{LHS of \eqref{eqn:taylor-score}} \le x) -   P^*( \text{RHS of \eqref{eqn:taylor-score}} \le x) | \le 1 - P^*(E_n).
\end{equation}
Thus, the proof requires studying the limit of the score $n^{-1/2}S_n(\theta^*)$, the limit of the derivative $n^{-1}S_n'(\theta^*)$, and verifying the remainder vanishes.

\paragraph{Step 1: CLT for the score at $\theta^*$}
The score evaluated at the true parameter is
\[
   S_n(\theta^*) = \sum_{t=1}^n \frac{\lambda_t'(\theta^*)(X_t - \theta^*) - \lambda_t(\theta^*)}{1 + \lambda_t(\theta^*)(X_t - \theta^*)} =: \sum_{t=1}^n Y_{t}.
\]
Using the identity $(1+u)^{-1} = 1 - u + \frac{u^2}{1+u}$, we decompose the summand $Y_{t}$. Let $Z_t := X_t - \theta^*$, $u_t := \lambda_t(\theta^*)Z_t$, and $N_t := \lambda_t'(\theta^*)Z_t - \lambda_t(\theta^*)$. Then:
\begin{align*}
  Y_{t} &= \frac{N_t}{1+u_t} = N_t (1 - u_t) + N_t \frac{u_t^2}{1+u_t} \\
  &= \lambda_t'(\theta^*)Z_t - \lambda_t(\theta^*) - \lambda_t'(\theta^*)\lambda_t(\theta^*)Z_t^2 + \lambda_t(\theta^*)^2 Z_t + N_t \frac{u_t^2}{1+u_t} \\
  &=: Y_{1t} + Y_{2t} + Y_{3t} + Y_{4t} + Y_{5t}.
\end{align*}
Let us analyze the sum of each term above.
\begin{enumerate}
    \item $\sum Y_{1t} = -\beta \sum Z_t + \sum (\lambda_t'(\theta^*) + \beta)Z_t$. The second sum is a martingale whose predictable quadratic variation is $\sigma^2 \sum (\lambda_t'(\theta^*) + \beta)^2$. By (A1), $\lambda_t'(\theta^*) \xrightarrow{L^2} -\beta$, so the variance is $o_p(n)$, implying that the martingale is $o_p(\sqrt{n})$. To conclude: $$\sum Y_{1t} = -\beta \sum Z_t + o_p(\sqrt{n}).$$
    \item $\sum Y_{2t} = -\sum \lambda_t(\theta^*) = -\beta \sum \bar Z_{t-1} - \sum r_t$. By (A2), $\sum |r_t| = o_p(\sqrt{n})$, yielding $$\sum Y_{2t} = -\beta \sum \bar Z_{t-1} + o_p(\sqrt{n}).$$
    \item $\sum Y_{3t} = \sum -\lambda_t'(\theta^*)\lambda_t(\theta^*) Z_t^2$ and $\Exp_{t-1} Y_{3t} = -\lambda_t'(\theta^*)\lambda_t(\theta^*) \sigma^2$. Let us analyze $\sum \lambda_t'(\theta^*)\lambda_t(\theta^*)$ and $\sum \lambda_t'(\theta^*)\lambda_t(\theta^*) (Z_t^2 - \sigma^2)$.
    \begin{enumerate}
        \item $\sum \lambda_t'(\theta^*)\lambda_t(\theta^*) = \sum (\lambda_t'(\theta^*) + \beta)\lambda_t(\theta^*) - \sum \beta \lambda_t(\theta^*)$. The first sum here is $o_p(\sqrt{n})$ by the same argument as the martingale part in Term 1 above. The second sum here is $\beta^2 \bar Z_{t-1} +  o_p(\sqrt{n})$ by (A2).
        \item $\sum \lambda_t'(\theta^*)\lambda_t(\theta^*)(Z_t^2 - \sigma^2) = \sum (\lambda_t'(\theta^*) + \beta)\lambda_t(\theta^*)(Z_t^2 - \sigma^2) - \sum \beta \lambda_t(\theta^*)(Z_t^2 - \sigma^2)$. The first sum here is $o_p(\sqrt{n})$ by the same argument as the martingale part in Term 1 above. For the second part, it is a martingale with predictable quadratic variation $O_{a.s.}(\sum \lambda_t^2(\theta^*)) \stackrel{\text{(A2)}}= O_p(\sum \bar Z_t^2) = o_p(\sqrt{n})$ so the martingale itself is also $o_p(\sqrt{n})$.
    \end{enumerate}
    Therefore, $$\sum Y_{3t} = \beta^2 \sigma^2  \bar Z_{t-1} + o_p(\sqrt{n}).$$

    \item $\sum Y_{4t} = \sum \lambda_t^2(\theta^*) Z_t$ is a martingale with predictable quadratic variation $\sum \lambda_t^4(\theta^*) \lesssim \sum \lambda_t^2(\theta^*) = o_p(\sqrt{n})$ again, so
    $$\sum Y_{4t} = o_p(\sqrt{n}).$$

    \item $\sum Y_{5t} = \sum N_t \frac{u_t^2}{1+u_t} \stackrel{\text{(A3)}}\le \sum N_t \frac{u_t^2}{\epsilon}  \stackrel{\text{(A3)}}\le \sum \frac{2L}{\epsilon}   u_t^2 \le  \frac{2L}{\epsilon} \sum \lambda^2_t(\theta^*)$. Recalling the sample mean approximation (A2) of $\lambda_t(\theta^*)$, we see that
    $$ \sum Y_{5t} = o_p(\sqrt{n}). $$
\end{enumerate}
Combining these bounds, the score satisfies:
\[
   S_n(\theta^*) = \sum_{t=1}^n \Bigl( -\beta Z_t - \beta \bar Z_{t-1} + \beta^2 \sigma^2 \bar Z_{t-1} \Bigr) + o_p(\sqrt{n}) = \sum_{t=1}^n \Bigl( -\beta Z_t - \rho \bar Z_{t-1} \Bigr) + o_p(\sqrt{n})
\]
where $\rho := \beta(1 - \beta\sigma^2)$. Note that the ``sum of sample mean'' part is a weighted, triangular array sum
\[
   \sum_{t=1}^n \bar Z_{t-1} = \sum_{t=2}^n \frac{1}{t-1} \sum_{i=1}^{t-1} Z_i = \sum_{i=1}^{n-1} Z_i \sum_{t=i+1}^n \frac{1}{t-1} = \sum_{i=1}^n Z_i \log \frac{n}{i} + \sum_{i=1}^n Z_i \left(H_n - H_i - \log \frac{n}{i} \right),
\]
where $ H_{n} = \sum_{t=1}^{n-1} 1/t$. The error term $\sum Z_i \left(H_n - H_i - \log \frac{n}{i} \right)$ above is a martingale with predictable quadratic variation that satisfies
\begin{align*}
   & \sum_{i=1}^n \left(H_n - H_i - \log \frac{n}{i} \right)^2 \le \sum_{i=1}^n \left((H_n - \log n) - (H_i - \log i)   \right)^2 \\
    \le& \sum_{i=1}^n \left(\frac{1}{2n} + \frac{1}{8n^2} + \frac{1}{2i} + \frac{1}{8i^2}\right)^2 = O(1),
\end{align*}
where the second inequality is due to $H_{n+1} = \log n + \gamma + 1/2n - \varepsilon_n $ where $0 \le \varepsilon_n \le 1/8n^2$ and $\gamma$ is the Euler–Mascheroni constant. This shows that $\sum Z_i \left(H_n - H_i - \log \frac{n}{i} \right) = o_p(\sqrt n)$. We thus see that
\begin{equation}
    S_n(\theta^*) = \sum_{i=1}^n Z_i \left( -\beta - \rho \log \frac{n}{i} \right) + o_p(\sqrt{n}),
\end{equation}
and the sum part's asymptotic variance can be approximated by an integral
\begin{align*}  
V_{\text{score}} &= \lim_{n\to\infty} \frac{1}{n} \sum_{i=1}^n \sigma^2 \left( -\beta - \rho \log\frac{n}{i} \right)^2 
    = \sigma^2 \int_0^1 \bigl( \beta - \rho \log x \bigr)^2 \d x \\
    &= \sigma^2 \left( \beta^2 - 2\beta\rho \int_0^1 \log x \d x + \rho^2 \int_0^1 (\log x)^2 \d x \right) \\
    &= \sigma^2 \bigl( \beta^2 + 2\beta\rho + 2\rho^2 \bigr).
\end{align*}
It is clear that the triangular array $ Z_i \left( -\beta - \rho \log \frac{n}{i} \right)$ satisfies the Lyapunov condition. Therefore, via the Lindeberg-Feller CLT, $n^{-1/2}S_n(\theta^*)$ converges in distribution to $\mathcal N(0, V_{\text{score}})$, where, recalling $\rho = \beta(1 - \beta\sigma^2)$,
\[
    V_{\text{score}} = \sigma^2 \bigl( \beta^2 + 2\beta\rho + 2\rho^2 \bigr) = \beta^2\sigma^2 \bigl( 5 - 6\beta\sigma^2 + 2\beta^2\sigma^4 \bigr).
\]

\paragraph{Step 2: LLN for the score derivative at $\theta^*$}
The derivative of the score summand is:
\[
S_n(\theta)  = \sum_{t=1}^n \frac{N_t(\theta)}{D_t(\theta)}, \quad  \left[ \frac{N_t(\theta)}{D_t(\theta)} \right]' = \frac{N_t'(\theta)D_t(\theta) - N_t(\theta)D_t'(\theta)}{D_t(\theta)^2}
\]
where $N_t(\theta) = \lambda_t'(\theta)(X_t - \theta) - \lambda_t(\theta)$ and $D_t(\theta) = 1 + \lambda_t(\theta)(X_t - \theta)$. 
Notice that $D_t'(\theta) = \lambda_t'(\theta)(X_t - \theta) - \lambda_t(\theta) = N_t(\theta)$. Thus, evaluating at $\theta^*$, the average can be decomposed into:
\begin{align*}
    S_n'(\theta^*) =\sum_{t=1}^n\frac{N_t'(\theta^*)}{D_t(\theta^*)} - \sum_{t=1}^n\frac{N_t(\theta^*)^2}{D_t(\theta^*)^2} =: \sum_{t=1}^n W_{1t} - \sum_{t=1}^n W_{2t}. 
\end{align*}
Let us analyze these two sums.

\begin{enumerate}
    \item $\sum W_{1t}$ can be further broken down into several parts. To do this, we use the identity $(1+u)^{-1} = 1 - u + \frac{u^2}{1+u}$ again and we have
    \[
    W_{1t} =  \frac{N_t'(\theta^*)}{1 +\lambda_t(\theta^*)Z_t } = N_t'(\theta^*) \left(1 - \lambda_t(\theta^*)Z_t + \frac{u_t^2}{1+u_t}\right),
    \]
    where $u_t = \lambda_t(\theta^*)Z_t$. By $N_t'(\theta^*) = \lambda_t''(\theta^*)Z_t - 2\lambda_t'(\theta^*)$, we have:
    \begin{align*}
        W_{1t} = & -2\lambda_t'(\theta^*) + \lambda_t''(\theta^*)Z_t + 2\lambda_t'(\theta^*) \lambda_t(\theta^*)Z_t - \lambda_t''(\theta^*) \lambda_t(\theta^*)Z_t^2 + \frac{u_t^2}{1+u_t} N_t'(\theta^*) \\
        =: & W_{11t} + W_{12t} + W_{13t} + W_{14t} + W_{15t}.
    \end{align*}
    
\begin{enumerate}
    \item $\frac{1}{n}\sum W_{11t} = \frac{1}{n}\sum (-2\lambda_t'(\theta^*)) \xrightarrow{p} 2\beta$ by (A1).
    \item $\frac{1}{n} \sum W_{12t} = \frac{1}{n}\sum \lambda_t''(\theta^*)Z_t \xrightarrow{p} 0$ because $\sum W_{12t}$ is a martingale with predictable quadratic variation $\sigma^2 \sum_{t=1}^n |\lambda_t''(\theta^*)|^2$ that is linearly bounded by (A3), so $\sum W_{12t} = o_p(\sqrt{n})$.
    \item $\frac 1 n \sum W_{13t} = \frac{1}{n} \sum 2\lambda_t'(\theta^*)\lambda_t(\theta^*)Z_t \xrightarrow{p} 0$ for the same reason.
    \item $\frac 1 n \sum W_{14t} =  \frac{1}{n}\sum \bigl(-\lambda_t''(\theta^*)\lambda_t(\theta^*)Z_t^2\bigr) \xrightarrow{p} 0$ because the boundedness of $\lambda_t''(\theta^*)$ from (A3), boundedness of $Z_t^2$, and the sample average approximation of $\lambda_t(\theta^*)$ from (A2) jointly imply that $\sum \bigl(-\lambda_t''(\theta^*)\lambda_t(\theta^*)Z_t^2\bigr) = O_p(\sqrt{n})$.
    \item $\frac 1 n \sum W_{15t} = \frac{1}{n}\sum \frac{u_t^2}{1+u_t}N_t'(\theta^*) \xrightarrow{p} 0$ because $\sum W_{15t} = o_p(\sqrt{n})$, similar to Term 5 in Step 1.
\end{enumerate}
The breakdown above shows that
\[
\frac{1}{n} \sum W_{1t} \xrightarrow{p} 2\beta.
\]
\item $\sum W_{2t}$ can be further broken down into several parts as well. To do this, we use the fact that there exists $K > 0$ such that $\left |(1+u)^{-2} - 1\right| \le  K|u|$ for all $u>- 1 + \epsilon$. So
    \[
    W_{2t} =  \frac{N_t^2(\theta^*)}{(1 +\lambda_t(\theta^*)Z_t)^2 } = N_t^2(\theta^*) \left(1  + R( u_t) \right), \quad |R(u_t^2)| \le K |u_t|
    \]
    where $u_t = \lambda_t(\theta^*)Z_t$ satisfies $u_t \ge -1+\epsilon$ by (A3). With
$N_t^2(\theta^*) = \lambda_t'(\theta^*)^2 Z_t^2 - 2\lambda_t'(\theta^*)\lambda_t(\theta^*)Z_t + \lambda_t^2(\theta^*)$, we have
\begin{align*}
    W_{2t} &= \lambda_t'(\theta^*)^2 Z_t^2 - 2\lambda_t'(\theta^*)\lambda_t(\theta^*)Z_t + \lambda_t^2(\theta^*) + N_t^2(\theta^*) R(u_t) \\
    &=: W_{21t} + W_{22t} + W_{23t} + W_{24t}.
\end{align*}
\begin{enumerate}
    \item $\frac{1}{n} \sum W_{21t} = \frac{1}{n}\sum \lambda_t'(\theta^*)^2 Z_t^2 \xrightarrow{p} \beta^2\sigma^2$ by $Z_t^2$'s SLLN and (A1).
    \item $\frac 1 n \sum W_{22t} = - \frac 1 n \sum W_{13t} \xrightarrow{p} 0$.
    \item $\frac 1 n \sum W_{23t} = \frac{1}{n}\sum_{t=1}^n \lambda_t^2(\theta^*) \xrightarrow{p} 0$ by (A2).
    \item $\frac 1 n \sum W_{24t} \xrightarrow{p} 0$ due to $|W_{24t}| \le K N_t^2(\theta^*) u_t^2 \stackrel{\text{(A3)}}\le 4K L^2 u_t^2 \le 4K L^2 \lambda_t^2(\theta^*) $ and (A2).
\end{enumerate}
The breakdown above shows that
\[
\frac{1}{n} \sum W_{2t} \xrightarrow{p} \beta^2 \sigma^2.
\]
\end{enumerate}

Combining both parts, we establish the WLLN limit:
\[
    \frac{1}{n}S_n'(\theta^*) \xrightarrow{p} 2\beta - \beta^2 \sigma^2.
\]

\paragraph{Step 3: Remainder}
By assumption (A3), $S_n''(\theta)/n$ is uniformly bounded in probability near $\theta^*$. Given $\tilde\theta_n^{(E)} \xrightarrow{p} \theta^*$, the second-order Taylor remainder $\frac{1}{2n}S_n''(\bar\theta_n)(\tilde\theta_n^{(E)} - \theta^*)$ is $o_p(1)$ and vanishes.

\vspace{.5em}

Assembling the three pieces into \eqref{eqn:taylor-score}, we apply Slutsky's theorem:
\[
   \sqrt{n}\bigl(\tilde\theta_n^{(E)} - \theta^*\bigr) \xrightarrow{d} \mathcal{N}\left(0, \frac{V_{\text{score}}}{(2\beta - \beta^2\sigma^2)^2}\right).
\]
Substituting $V_{\text{score}} =  \beta^2\sigma^2 \bigl( 5 - 6\beta\sigma^2 + 2\beta^2\sigma^4 \bigr)$, the proof is complete.
\end{proof}

\subsection{Proof of \cref{thm:saddle}} \label[appendix]{sec:pf-saddle}

\begin{proof}
    Let the objective function be denoted by 
    \[ f(\theta, \lambda) = \sum_{t=1}^n \log(1 + \lambda(X_t - \theta)), \]
    and let the permissible domain for $\lambda$ given $\theta$ be $\Lambda(\theta) = [-\frac{1}{1-\theta}, \frac{1}{\theta}]$. We wish to find $\argmin_{\theta \in [0,1]} \sup_{\lambda \in \Lambda(\theta)} f(\theta, \lambda)$.

    First, observe that for any candidate parameter $\theta \in [0,1]$, the choice $\lambda = 0$ is always valid (i.e., $0 \in \Lambda(\theta)$). Evaluating the function at $\lambda = 0$ gives:
    \[ f(\theta, 0) = \sum_{t=1}^n \log(1 + 0) = 0. \]
    Therefore, the inner supremum is always bounded below by $0$:
    \begin{equation} \label{eqn:saddle_lower_bound}
        \sup_{\lambda \in \Lambda(\theta)} f(\theta, \lambda) \ge f(\theta, 0) = 0 \quad \text{for all } \theta \in [0,1].
    \end{equation}

    Now, let us evaluate the inner supremum at the specific choice $\theta = \bar{X}_n$. Because all observations $X_t \in [0,1]$, their empirical mean $\bar{X}_n$ is also in $[0,1]$. We consider the function of $\lambda$ at this point:
    \[ h(\lambda) := f(\bar{X}_n, \lambda) = \sum_{t=1}^n \log(1 + \lambda(X_t - \bar{X}_n)). \]
    
    The function $h(\lambda)$ is strictly concave in $\lambda$ over its domain (unless all $X_t$ are identical, in which case $h(\lambda) = 0$ for all $\lambda$, trivially achieving a supremum of $0$). We compute the first derivative of $h(\lambda)$ with respect to $\lambda$:
    \[ h'(\lambda) = \sum_{t=1}^n \frac{X_t - \bar{X}_n}{1 + \lambda(X_t - \bar{X}_n)}. \]
    
    Evaluating this derivative at $\lambda = 0$ yields:
    \[ h'(0) = \sum_{t=1}^n (X_t - \bar{X}_n) = \sum_{t=1}^n X_t - n \bar{X}_n = 0. \]
    
    Because $h(\lambda)$ is a concave function, a critical point where the first derivative is zero must be a global maximum. Consequently, the maximum of $h(\lambda)$ over $\Lambda(\bar{X}_n)$ is achieved at $\lambda = 0$:
    \begin{equation} \label{eqn:saddle_mean_eval}
        \sup_{\lambda \in \Lambda(\bar{X}_n)} f(\bar{X}_n, \lambda) = h(0) = 0.
    \end{equation}

    Combining \eqref{eqn:saddle_lower_bound} and \eqref{eqn:saddle_mean_eval}, we see that the function $\theta \mapsto \sup_{\lambda} f(\theta, \lambda)$ attains its global minimum of $0$ at $\theta = \bar{X}_n$. Thus, the saddle point ME-estimator is precisely the sample mean.
\end{proof}

\end{document}

%% file: macros.tex
\renewcommand{\le}{\leqslant}

\renewcommand{\ge}{\geqslant}

\newcommand{\id}{\mathbbmss 1}

\newcommand {\matl}{\left[ \begin{matrix}}
\newcommand {\matr}{\end{matrix}\right]}
\newcommand {\Exp}{ \mathbb E }

\newcommand {\sg}{\operatorname{sgn}}
\renewcommand {\Pr}{ \mathbb P }

\newcommand{\blamb}{\boldsymbol{\lambda}}

\renewcommand{\d}{\mathrm{d}}
\newcommand{\iid}{\overset{\mathrm{iid}}{\sim}}

\newcommand{\cP}{\mathcal{P}}

\newcommand{\CS}{\operatorname{CS}}
\newcommand{\CI}{\operatorname{CI}}
\newcommand{\diam}{\operatorname{diam}}

\DeclareMathAlphabet{\mathbbmsl}{U}{bbm}{m}{sl}

\DeclareMathOperator*{\argmin}{arg\,min}

\newcommand{\tr}{\operatorname{tr}}

\newcommand{\klinf}{\operatorname{KL}_{\inf}}